\documentclass[journal]{IEEEtran}      

\IEEEoverridecommandlockouts                              

\usepackage{amsmath} 
\usepackage{amssymb} 
\usepackage{amsthm}  
\usepackage{cite,enumerate}
\usepackage[pdftex]{graphicx}
\usepackage{color}
\usepackage{subfigure}
\usepackage{comment}
\usepackage{url}
\usepackage{algorithm,algorithmic}

\newtheorem{thm}{\bf Theorem}
\newtheorem{lem}{\bf Lemma}

\newtheorem{rem}{\bf Remark}
\newtheorem{assmpt}{\bf Assumption}

\newcommand{\bG}{\ensuremath{{\mathbb G}}}

\newcommand{\bF}{\ensuremath{{\mathbb F}}}
\newcommand{\bR}{\ensuremath{{\mathbb R}}}
\newcommand{\bN}{\ensuremath{{\mathbb N}}}
\newcommand{\bX}{\ensuremath{{\mathbb X}}}
\newcommand{\bY}{\ensuremath{{\mathbb Y}}}

\newcommand{\ep}{\ensuremath{\epsilon}}

\newcommand{\cB}{\ensuremath{\mathcal{B}}}
\newcommand{\cC}{\ensuremath{\mathcal{C}}}

\newcommand{\cG}{\ensuremath{\mathcal{G}}}

\newcommand{\cL}{\ensuremath{\mathcal{L}}}

\newcommand{\cN}{\ensuremath{\mathcal{N}}}

\newcommand{\col}{\ensuremath{{\rm col}}}
\newcommand{\row}{\ensuremath{{\rm row}}}
\newcommand{\diag}{\ensuremath{{\rm diag}}}

\allowdisplaybreaks[3] 

\title{Decentralized Design and 
Plug-and-Play Distributed Control for Linear Multi-Channel Systems
}

\author{Taekyoo Kim, Donggil Lee, and Hyungbo Shim,~\IEEEmembership{Senior~Member,~IEEE }
\thanks{This work was supported 
	by the BK21 FOUR program of the Education and Research Program for Future ICT Pioneers, and
by the National Research Foundation of Korea (NRF) grant funded by the Korea government (MSIT) (No. NRF-2017R1E1A1A03070342).
}
\thanks{Taekyoo Kim, Donggil Lee, and Hyungbo Shim
	are with ASRI, Department of Electrical and Computer Engineering, Seoul National University, Korea{\tt\small (taekyoo.kim@cdsl.kr, dglee@cdsl.kr, hshim@snu.ac.kr)}.}%

}

\allowdisplaybreaks[3]

\begin{document}

\maketitle
\thispagestyle{plain} 
\pagestyle{plain}

\begin{abstract} 
We propose a distributed control, in which many identical control agents are deployed for controlling a linear time-invariant plant that has multiple input-output channels.
Each control agent can join or leave the control loop during the operation of stabilization without particular initialization over the whole networked agents.
Once new control agents join the loop, they self-organize their control dynamics, which does not interfere the control by other active agents, which is achieved by local communication with the neighboring agents.
The key idea enabling these features is the use of Bass' algorithm, which allows the distributed computation of stabilizing gains by solving a Lyapunov equation in a distributed manner.
\end{abstract}


\section{Introduction}

Distributed control is receiving a lot of attention in response to the recent demand for controlling a dynamic system via spatially deployed multi-agents (i.e., networked local controllers).
This paper continues the work initiated by \cite{Wang2017,LWang20} in this regard.
In particular, we consider a multi-channel linear time-invariant plant written as
\begin{subequations}\label{eq:plant}
	\begin{align}
		\dot x(t) &= A x(t) + \sum_{i \in \cN(t)} B_i u_i(t)  \label{eq:plant_in} \\
		y_i(t) &= C_i x(t), \qquad i \in \cN(t)  \label{eq:plant_out}
	\end{align}
\end{subequations}
where $x\in\bR^n$ is the plant state, $u_i\in\bR^{m_i}$ and $y_i\in\bR^{p_i}$ are the input and the output corresponding to the $i$th channel, respectively.
The plant is controlled by {\em networked control agents} via input-output channels; that is, a control agent in charge of the $i$th channel measures the output $y_i$, has access to the input $u_i$, and communicates through a bidirectional communication network with its neighboring control agents.
We suppose there are at most $\bar N$ channels, some of which are active and the rest are idle.
Active channels imply there are control agents that are using the channel, and the idle channels do not have corresponding control agents.
Status of the channel between active and idle can vary as time goes on.
The index set of active channels at time $t$ is denoted by $\cN(t) \subset \{1,\dots,\bar N\}$.
The situation is depicted in Fig.~\ref{fig:1}.

We will design {\em identical} control agents that collectively stabilize the plant \eqref{eq:plant}; that is, a control agent does not have its own designated channel, and when a control agent happens to link to any channel, it automatically designs its own control gains and takes part in stabilization in harmony with other active agents.
It is assumed that an agent learns $A$, $B_i$, and $C_i$ when it links to the $i$th channel, but they cannot learn $B_j$ and $C_j$ for $j \not = i$.
For convenience of forthcoming derivation, we assume that 
\begin{align}\label{eq:B_i_C_i}
	\|B_i\| \le 1, \quad \|C_i\| \le 1, \qquad \forall i.
\end{align}
This is indeed no loss of generality because when an agent gets to know $B_i$ and $C_i$, the agent can redefine them as $\tilde B_i = B_i/\|B_i\|$ and $\tilde C_i = C_i/\|C_i\|$, and treat $\tilde y_i = y_i/\|C_i\|$ as the output and $\tilde u_i = \|B_i\|u_i$ as the input for its own use.

\begin{figure}
	\centering
	\includegraphics[width=1\columnwidth]{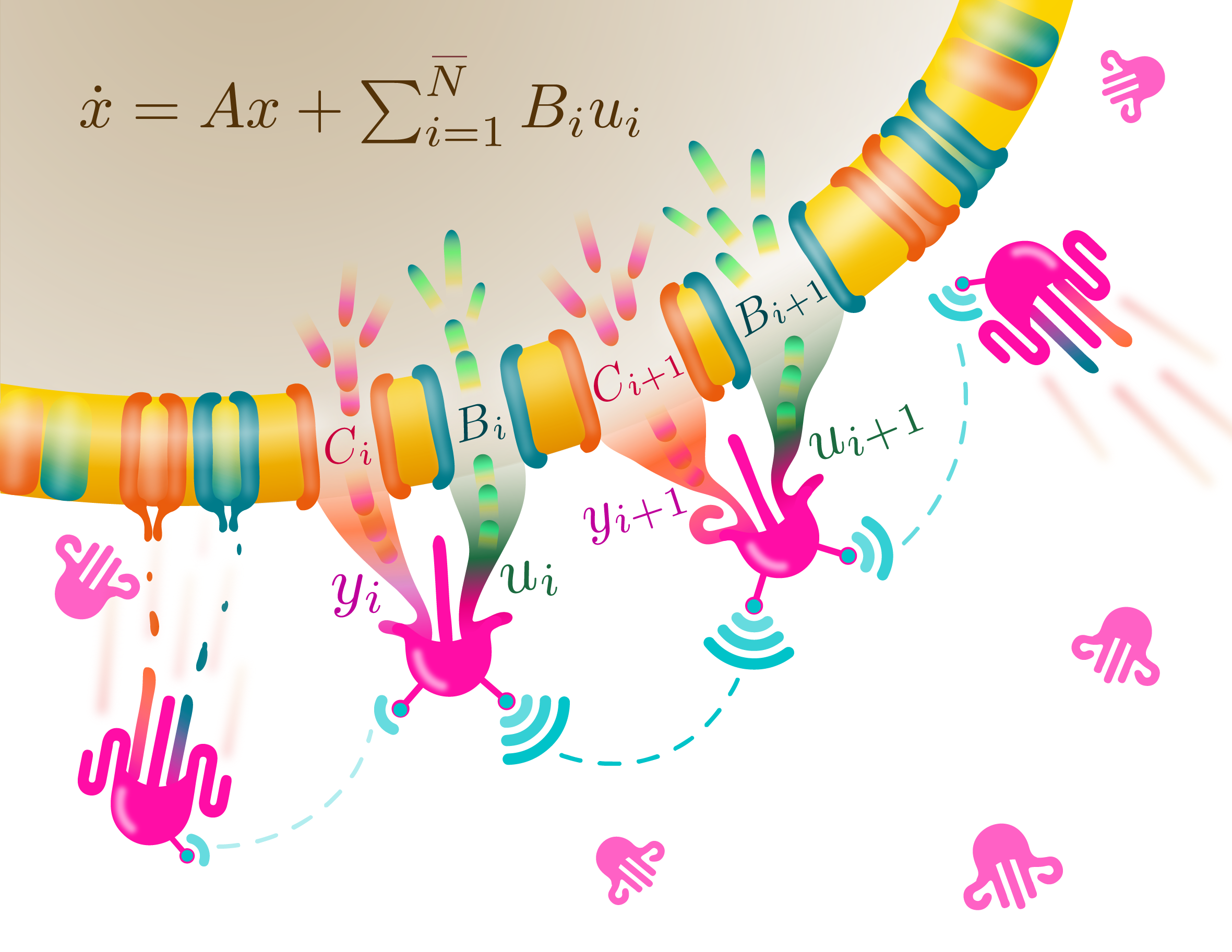}
	\caption{The plant (in brown color) is under control by multiple control agents (in pink color) that can freely join or leave the control network. Each agent gets to know $A$, $B_i$, and $C_i$ (but not $B_j$ and $C_j$, $j \not = i$), obtains the measurement $y_i$ (but not $y_j$, $j \not = i$), and generates control action $u_i$ based on local communication (dashed line in sky color) with its nearby neighbors. No centralized coordinator exists and each agent self-organizes its own dynamics for control action.}
	\label{fig:1}
\end{figure}

To achieve the goal of stabilization, we assume that $\cN(t)$ does not change too often (whose meaning will be clarified in Section \ref{sec:dwell}), that the plant with active channels is collectively controllable and observable [i.e., $(A,\{B_i, i\in\cN(t)\}, \{C_i,i\in\cN(t)\})$ is controllable and observable for each $t$], and that the network graph describing the communication among the active agents is undirected and connected for all time.
The problem setup is well-suited for cooperative multi-robots, and an example is illustrated in Section \ref{sec:casestudy}.

Key features of the proposed distributed control include the following:
\begin{enumerate}

\item [(F1)] (Distributed operation) Each control agent interacts with the plant through its local input-output channel and exchanges information with its neighboring agents through local communication. No central unit exists.

\item [(F2)] (Decentralized design) Each agent {\em self-organizes} its own control dynamics from its local knowledge.

\item [(F3)] (Plug-and-play) Each agent can join or leave the network during the operation without resetting or re-initializing the states of all other agents in the network.

\end{enumerate}

The feature of plug-and-play is challenging but pursued in the literature because it yields flexibility and resiliency, which are essential in networked control system, especially for large-scale systems with multiple channels \cite{Riverso2013,Bendtsen2013,Zeilinger2013}. 
For instance, when there is a change in the network configuration, the plug-and-play capability implies that the functionality will be automatically recovered soon after the change even if not all agents notice the change and do anything special in response to the change. 
Moreover, combined with the capability of decentralized design, implementation and maintenance of the networked control system become simplified. 
For example, when a malfunction is observed, one may simply deploy new control agents without identifying the fault (since all the control agents are identical and they learn their role online) and without stopping the system (thanks to the plug-and-play feature).

In Section \ref{sec:2}, we develop an algorithm for distributed output-feedback stabilization, which is embedded in each control agent, and exhibits the features (F1), (F2), and (F3).


\subsection{Literature survey and the proposed contribution}

A related problem has been extensively studied in the literature under the name of `decentralized control' around 1970s (e.g., see \cite{Siljack2011} and references therein).
A difference is that no communication is allowed between controllers, which poses a structural constraint on the control design.
It is shown in \cite{SWnag1973} that the decentralized controller exists if and only if the plant's `fixed spectrum' is stable, which is a set of closed-loop eigenvalues invariant to changes of local controllers.

By introducing inter-controller communications, a more general class of plants can be handled by so-called `distributed control.'
The first step towards distributed output-feedback control was, motivated by the classical separation principle, to employ {\em distributed observers} studied in \cite{Zhu2014,Park2017,Wang2017,Mitra2018,Han2018,TKim2020}.
However, while these methods allow the agent $i$ to distributedly compute the estimate $\hat x_i$ of the plant's state $x$, a difficulty arises unless the input signals to the plant are shared by all agents.
In fact, a naive combination of the distributed observer and a feedback control $u_i = F_i \hat x_i$ violates both the distributed operation (F1) and the decentralized design (F2) because the input to the plant is now $\sum_j B_j u_j = \sum_j B_j F_j \hat x_j$, and therefore, the distributed observer in the control agent $i$ needs to know all other $B_j$, $F_j$, and $\hat x_j$, $j \not = i$, as well, in order to cancel them in its own observer error dynamics. 
Recently, the authors of \cite{LWang20} and \cite{XZhang2019} have proposed a solution, which enabled that agent $i$ uses $(\sum_j B_j F_j) \hat x_i$ instead of $\sum_j B_j F_j \hat x_j$, so that the real-time information of $\hat x_j$ ($j \not = i$) is not needed for agent $i$.
The underlying idea is that, when the consensus is achieved among all the estimates $\hat x_j$, both terms become the same.
Hence, the distributed operation (F1) is achieved in \cite{XZhang2019,LWang20}.

Nevertheless, achieving all three goals (F1), (F2), and (F3) is still challenging.
The controller design in \cite{XZhang2019,LWang20} is centralized in the sense that the observer in the agent $i$ still needs other $B_j$ and $F_j$, and moreover, design of its own gain $F_i$ involves non-local information such as other agents' $B_j$ and/or gains $F_j$. 
Our first contribution is the decentralized design (F2) of $F_i$ (and the gains of distributed observers as well).
It will turn out that this is possible thanks to a distributed computation of a Lyapunov equation, whose idea is inspired by the blended dynamics approach \cite{JKim2016,JLee2020}.
Specifically, we synthesize the identical control agent, which can self-organize its own control gains from its locally accessible information such as $(A,B_i,C_i)$.
The proposed distributed control algorithm also supports the plug-and-play operation (F3), as long as controllability and observability are maintained and the communication graph remains connected.
In fact, the plug-and-play feature enables stabilization of the plant when the network change is not too often.
Since the term `not too often' is vague, we obtain in Section \ref{sec:dwell} a dwell time of the change such that the closed-loop system remains stable even if any control agents join or leave the network, as long as each change occurs at least after the dwell time.

\subsection{Notation and preliminaries}

We let $1_N$ be the $N \times 1$ column vector comprising all ones, and denote by $I_N$ the $N \times N$ identity matrix.
For matrices $A$ and $B$, their Kronecker product is denoted by $A\otimes B$, and for $A \in \bR^{m \times m}$ and $B \in \bR^{n \times n}$, their Kronecker sum $A\otimes I_n +I_m \otimes B$ is denoted by $A \oplus B$. 
For matrices $A$ and $B$, we use the notation $\row(A,B) := [A,B]$ and $\col(A,B) := [A^T, B^T]^T$ as long as they are well-defined, and the block-diagonal matrix of $A$ and $B$ is simply written as $\diag(A,B)$.
For a vector $x$ and a matrix $P$, $|x|$, $\|P\|$, and $\|P\|_F$ denote the Euclidean norm, the induced matrix 2-norm, and the Frobenius norm, respectively. For a matrix $P\in\bR^{n\times n}$, ${\rm vec}(P)$ denotes its vectorization (a column stack of all the columns of $P$) in $\bR^{n^2}$, which satisfies that $\|P\| \le |{\rm vec}(P)|=\|P\|_F \le \sqrt{n}\|P\|$.
For matrices $A$ and $B$ of any size, $\|A \otimes B\| = \|A\| \|B\|$.

A graph is denoted by $\cG$ with the set of its nodes as $\cN$.
We use standard terminology for graphs as in \cite{Mesbahi2010}, and consider {\em undirected} and {\em connected} graphs in this paper.
A graph is {\em unweighted} if all components of its {\em adjacency matrix} have either 0 or 1. 
With the (symmetric) {\em Laplacian matrix} $\cL$ of a undirected graph $\cG$, we denote the eigenvalues of $\cL\in\bR^{|\cN| \times |\cN|}$ by $\lambda_1(\cL),\lambda_2(\cL),\cdots,\lambda_{|\cN|}(\cL)$, where $\lambda_1(\cL)=0$ and $\lambda_i(\cL)\le\lambda_j(\cL)$, for $i<j$, and $|\cN|$ is the cardinality of the set $\cN$.
A undirected connected graph has simple zero eigenvalue, or equivalently, $\lambda_2(\cL)>0$. 
In this case, there is a matrix $R\in\bR^{|\cN| \times (|\cN|-1)}$ such that
\begin{align}\label{eq:R}
	1_{|\cN|}^T R =0,\quad R^TR=I_{(|\cN|-1)},\quad R^T\cL R=\Lambda^+
\end{align}
where $\Lambda^+\in\bR^{(|\cN|-1)\times(|\cN|-1)}$ is the diagonal matrix with positive entries $\lambda_2(\cL)$, $\cdots$, $\lambda_{|\cN|}(\cL)$.
It is well known from \cite[Theorem 1]{Anderson85} that 
\begin{equation}\label{eq:Anderson85}
\|\cL\|\le |\cN|,
\end{equation}
and, from \cite{Mohar91} that for any connected unweighted undirected graph $\cG$
\begin{align}\label{eq:Mohar91}
	\lambda_2(\cL) \ge 4/{|\cN|}^2.
\end{align}


\section{Synthesis of Control Agents}\label{sec:2}


In this section, we present a design of control agents for distributed control that features three properties (F1), (F2), and (F3) in the Introduction, under the following assumption which will be relaxed in the following subsections.
\begin{assmpt}[Tentative assumption]\label{ass:tentative}
\begin{enumerate}
\item [(T1)] There are $N$ active agents, and the number $N$ is known to all agents.
\item [(T2)] The communication graph $\cG$ among the active agents is undirected, unweighted, and connected with $\cN = \{1,\dots,N\}$.
\item [(T3)] The plant $(A,B,C)$ is controllable and observable where 
\begin{align*}
B:=[B_1,B_2,\cdots,B_N], \quad C:=[C_1^T,C_2^T,\cdots,C_N^T]^T.
\end{align*}
\item [(T4)] Gain matrices $F_i \in \bR^{m_i \times n}$ and $L_i \in \bR^{n \times p_i}$ are designed such that
\begin{enumerate}
\item $A + \sum_{j=1}^N B_j F_j = A + BF$ is Hurwitz
\item $A + \sum_{j=1}^N L_j C_j = A + LC$ is Hurwitz
\end{enumerate}
where $F = [F_1^T,\cdots,F_N^T]^T$ and $L = [L_1,\cdots,L_N]$. 
Two gains $F_i$ and $L_i$ are known to the agent linked to the channel $i$.
\end{enumerate}
\end{assmpt}

\subsection{Distributed State Observer and State Feedback}

We design the control agent as a combination of distributed state observer and a state feedback.
In particular, the agent $i$ is given by
\begin{align}
\dot{\hat x}_i &= A \hat x_i + N B_i F_i \hat x_i + N L_i (C_i \hat x_i - y_i) + \gamma \sum_{j \in \cN_i} (\hat x_j-\hat x_i) \notag \\
u_i &= F_i \hat x_i \label{eq:prop}
\end{align}
where $\cN_i$ denotes the index set of the neighboring agents of agent $i$, and $\gamma \in \bR$ is the coupling gain to be designed. It is noted that the first two terms resemble the copy of the plant while the plant's input term $\sum_{j=1}^N B_j u_j = \sum_{j=1}^N B_j F_j \hat x_j$ is replaced with $NB_iF_i\hat x_i$.
The third term serves the standard injection of output error, which is inflated $N$ times.
The last term is the coupling that enables the communication with the neighbors, and is the main player that enforces synchronization of all $\hat x_i \in \bR^n$.

The intuition for the form \eqref{eq:prop} is rooted in the blended dynamics theorem \cite{JKim2016,JLee2020}, which can be roughly summarized as follows.
(See \cite{BlendedBook} for a comprehensive summary of this approach.)
For a network of heterogeneous dynamics
$$\dot x_i = f_i(x_i) + \gamma \sum_{j \in \cN_i}(x_j - x_i), \qquad i = 1,\dots,N,$$
under a undirected and connected graph, if all the states $x_i$ are enforced to synchronize (by sufficiently large $\gamma$), then a collective behavior emerges, which approximately obeys the so-called blended dynamics defined by
$$\dot s = \frac1N \sum_{i=1}^N f_i(s)$$
as long as the blended dynamics is stable.
For our case of \eqref{eq:prop}, the blended dynamics becomes
\begin{align}\label{eq:ourblended}
\begin{split}
	\dot s &= As + \sum_{i=1}^N B_i F_i s + \sum_{i=1}^N L_i C_i s - \sum_{i=1}^N L_i y_i \\
	&= As + BFs + L(Cs - y)
\end{split}
\end{align}
where $s \in \bR^n$ and $y = [y_1^T,\dots,y_N^T]^T$, which is the same as the standard, single, observer-based controller that makes the closed-loop system stable under the separation principle. However, unlike the blended dynamics theorems in \cite{JKim2016,JLee2020}, the blended dynamics \eqref{eq:ourblended} is not a stable system in general, and therefore, we have to work on a corresponding error dynamics as done in the proof of the following theorem.

\begin{thm}\label{thm:separation}
(i) Under Assumption \ref{ass:tentative}, the closed-loop consisting of the plant \eqref{eq:plant} and the distributed controllers \eqref{eq:prop} is an exponentially stable LTI system if the coupling gain $\gamma$ is greater than a threshold $\bar\gamma$, i.e., if $\gamma > \bar\gamma$.

(ii) The threshold $\bar\gamma$ is given by
\begin{align}
\bar\gamma = \frac{1}{\lambda_2(\cL)}\bigg(\theta+\theta^2\kappa
 +4N^2\max\limits_{i\in\cN}\|F_i\|^2\kappa\sqrt{1+\theta^2\kappa^2}\,\bigg) \label{eq:gamma_general_1}
\end{align}
where
\begin{itemize}
\item $\theta:= \|A\|+N\max\limits_{i\in\cN}\|L_i\|+2N\max\limits_{i\in\cN}\|F_i\|$, where $\cN=\{1,\dots,N\}$, and
\item $\kappa = \lambda_{\max}(\diag(M_1,M_2))/\lambda_{\min}(\diag(Q_1,Q_2))$
with positive definite $M_1$, $M_2$, $Q_1$, and $Q_2$ satisfying
\begin{align}\begin{split}\label{eq:sep_Lyap_eq}
		M_1(A+BF) + (A+BF)^TM_1 &= -2Q_1 \\
		M_2(A+LC) + (A+LC)^TM_2 &= -2Q_2.
\end{split}\end{align}
\end{itemize}
\end{thm}

It is seen from (i) that, when $\gamma$ is large enough, the closed-loop system becomes asymptotically stable, and the question how large $\gamma$ should be is answered by (ii).
While the value $\bar\gamma$ consists of global information, it will be shown that it can be obtained in a distributed way.

\begin{proof}
Define the error variable $e_i:=\hat x_i - x$.
Then, it holds that $\sum_{j \in \cN_i}(\hat x_j-\hat x_i) = \sum_{j \in \cN_i}(e_j-e_i) = -\sum_{j=1}^N l_{ij} e_j$ where $l_{ij}$ is the $(i,j)$-th entry of the graph Laplacian $\cL$.
Also, it is seen, with \eqref{eq:plant} and \eqref{eq:prop}, that the error dynamics becomes
\begin{align*}
	\dot e_i &= (A+NL_iC_i)e_i +(NB_iF_ie_i-\sum_{j=1}^N B_jF_je_j) \notag\\
	&\quad +(NB_iF_i-BF)x-\gamma\sum_{j=1}^N l_{ij}e_j, \qquad i = 1,\dots,N.
\end{align*}
The dynamics of the aggregated error $e:=\col_{i=1}^N(e_i)\in\bR^{Nn}$ is then written by
\begin{align}\label{eq:dot_e}
	\dot e &= G e+Hx-\gamma(\cL \otimes I_n)e
\end{align}
where $G\in\bR^{Nn \times Nn}$ and $H\in\bR^{Nn \times n}$ are defined by
\begin{align}\begin{split}\label{eq:G_H}
		G &:= \diag_{i=1}^N(A+NL_iC_i+NB_iF_i) \\
		&\qquad \qquad -1_N\otimes \begin{bmatrix} B_1F_1 &\cdots& B_NF_N \end{bmatrix} \\
		H &:= \col_{i=1}^N(NB_iF_i-BF).
\end{split}\end{align}
Let us decompose $e$ into its average $\bar e$ and the rest $\tilde e$ such that
\begin{align*}
\bar e:= (\frac{1}{N}1_N^T \otimes I_n)e \in \bR^n, \quad 
	\tilde e:= (R^T \otimes I_n)e \in \bR^{(N-1)n}
\end{align*}
where the matrix $R$ satisfies \eqref{eq:R}, which leads to
$$e=(1_N \otimes I_n)\bar e+(R \otimes I_n)\tilde e.$$
By applying this coordinate change to \eqref{eq:dot_e} and \eqref{eq:plant_in}, we have
\begin{align}\label{eq:dot_e_dot_x}
	\begin{bmatrix}
		\dot x \\
		\dot{\bar e} \\
		\dot{\tilde e}
	\end{bmatrix}
	\!\!=\!\!
	\begin{bmatrix}
		A+BF & BF & \square_1 \\
		0 & A+LC & \square_2 \\
		\square_3 & \square_4 & \square_5-\gamma( \Lambda^+\otimes I_n)
	\end{bmatrix}\!\!\!
	\begin{bmatrix}
		x \\
		\bar e \\
		\tilde e  
	\end{bmatrix}
\end{align}
where $\Lambda^+$ is given in \eqref{eq:R} and
\begin{align*}
\square_1 &:=\begin{bmatrix}
B_1F_1 & \cdots & B_NF_N
\end{bmatrix}(R\otimes I_n) \\
\square_2 &:=(\frac{1}{N}1_N^T\otimes I_n)\left[\diag_{i=1}^N(A+NL_iC_i)\right](R\otimes I_n) \\
\square_3 &:= (R^T\otimes I_n)H \\
\square_4 &:= (R^T\otimes I_n)G(1_N^T\otimes I_n) \\
\square_5 &:= (R^T\otimes I_n)G(R\otimes I_n).
\end{align*}
Here, as intended in \eqref{eq:prop}, we observe that the upper left $2 \times 2$ block of \eqref{eq:dot_e_dot_x} is the typical output-feedback configuration for the system $(A,B,C)$, which is Hurwitz.
In addition, $\Lambda^+$ is positive definite, and thus, a sufficiently large $\gamma$ yields stability of \eqref{eq:dot_e_dot_x} and the property (i).
More rigorous analysis using a Lyapunov function (which also yields the property (ii)) is found in the Appendix \ref{app:1}.
\end{proof}

\begin{rem}
The proposed control agents \eqref{eq:prop} stabilizes the plant \eqref{eq:plant} featuring the distributed operation (F1), and thus, serve as an alternative solution of \cite{XZhang2019,LWang20}.
In fact, the distributed controller proposed in \cite{XZhang2019,LWang20} is similar to \eqref{eq:prop} but different in that the term
\begin{equation}\label{eq:other}
\left(\sum_{j=1}^N B_jF_j\right) \hat x_i
\end{equation}
was used instead of the term $NB_iF_i \hat x_i$ in \eqref{eq:prop}.
In order to implement \eqref{eq:prop} with \eqref{eq:other}, each control agent needs to know $B_j$ and $F_j$ of other agents, which is a drawback.
Technically, employing \eqref{eq:other} in \eqref{eq:prop} renders $\square_3=0$ in \eqref{eq:dot_e_dot_x}, which is not very helpful in view of stabilization by sufficiently large $\gamma$.
\end{rem}

\subsection{Removing Assumption (T4): Distributed Computation of $F_i$ and $L_i$}

The distributed controller in the previous subsection does not meet the requirement of decentralized design (F2) because, in order to design $F_i$, one needs to find $F_i$ such that 
$$A+\begin{bmatrix} B_1 & \cdots & B_N \end{bmatrix} \begin{bmatrix} F_1 \\ \vdots \\ F_N \end{bmatrix} \quad \text{is Hurwitz}$$
which implies that the selection of $F_i$ is affected by all $B_j$'s and all other $F_j$'s in general.
In other words, selection of $F_i$'s is correlated with other $F_j$'s.
Our idea of handling this correlation is based on a Lyapunov equation, arising from the following Bass' Algorithm.

\begin{lem}[Bass' Algorithm \cite{Bass61}]\label{thm:Bass}
Assume that $(A,B)$ is controllable and let a positive scalar $\beta$ be such that $-(A+\beta I_n)$ is Hurwitz.
Then the algebraic Lyapunov equation
\begin{align}\label{eq:Bass}
-(A+\beta I_n)X-X(A+\beta I_n)^T+2BB^T=0
\end{align}
has a unique solution $X_*$, which is positive definite. 
Moreover, the matrix $(A+BF_*)$ is Hurwitz where 
\begin{align}\label{eq:gain_Bass}
F_* = -B^T X_*^{-1}
\end{align}
and the response of $\dot x = (A+BF_*) x$ satisfies
\begin{equation}\label{eq:converge}
|x(t)| \le \sqrt{\frac{\lambda_{\max}(X_*^{-1})}{\lambda_{\min}(X_*^{-1})}} e^{-\beta t} |x(0)|, \qquad \forall t \ge 0.
\end{equation}
\end{lem}

\begin{proof}
Existence of a unique positive definite solution $X_*$ is a standard result of algebraic Lyapunov equation.
For the rest, multiply the left and the right side of \eqref{eq:Bass} by $X_*^{-1}$, and obtain
\begin{align*}
X_*^{-1}(A-BB^TX_*^{-1})+(A-BB^TX_*^{-1})^TX_*^{-1}=-2\beta X_*^{-1}
\end{align*}
which proves that $(A+BF)$ is Hurwitz. 
Inequality \eqref{eq:converge} can be seen through the Lyapunov function $V(x) = x^TX_*^{-1}x$.
\end{proof}

By Lemma \ref{thm:Bass}, if $X_*$ is available to every agents, then each agent can compute
\begin{align}\label{eq:F_i_static}
	F_{*,i} = - B_i^T X_*^{-1}, \quad i = 1, \cdots, N
\end{align}
which does not need the knowledge of other $B_j$'s, where $F_* = \col_{i=1,\dots,N}(F_{*,i})$ by \eqref{eq:gain_Bass}.
However, the solution $X_*$ to \eqref{eq:Bass} contains information about all $B_j$'s.
So the question now is how to compute $X_*$ in a distributed manner.
To answer this question, we observe that $X_*$ can be asymptotically obtained by the differential Lyapunov equation:
\begin{align}\label{eq:difBass0}
\dot X = -(A+\beta I_n)X-X(A+\beta I_n)^T+2BB^T.
\end{align}
Here, a possible idea is to construct dynamics for agent $i$ (which has information of $A$ and $B_i$ only) as
\begin{multline}\label{eq:candidate0}
\dot X_i = -(A+\beta I_n)X_i-X_i(A+\beta I_n)^T +2NB_iB_i^T \\
+ \gamma \sum_{j \in \cN_i} (X_j - X_i),
\end{multline}
whose blended dynamics turns to be \eqref{eq:difBass0}.
However, one issue of \eqref{eq:candidate0} is that, while $X_*$ is the equilibrium of \eqref{eq:difBass0}, it is not an equilibrium for \eqref{eq:candidate0} (i.e., \eqref{eq:candidate0} does not hold with $X_i(t) \equiv X_*$, $\forall i \in \cN$).
Therefore, it is hopeless for \eqref{eq:candidate0} to have $X_i(t) \to X_*$ even if consensus of all $X_i$'s is enforced.

To solve this issue, we propose the following dynamics (instead of \eqref{eq:candidate0}) for each agent:
\begin{subequations}\label{eq:Z_i_X_i}
\begin{align}
\dot Z_i &= -\gamma_c \sum_{j\in\cN_i} (X_j-X_i) \label{eq:Z_i_X_i1} \\
\dot X_i &= k_c\left[-(A+\beta I_n) X_i - X_i (A+\beta I_n)^T + 2B_iB_i^T\right] \notag \\
&\qquad + \gamma_c \sum_{j\in\cN_i} (X_j-X_i) + \gamma_c \sum_{j\in\cN_i} (Z_j-Z_i) \label{eq:Z_i_X_i2}
\end{align}
\end{subequations}	
where $Z_i \in \bR^{n \times n}$ and $X_i \in \bR^{n \times n}$,
and $\beta$ is set as in Lemma~\ref{thm:Bass}.
We assume that all $X_i(0)$ and $Z_i(0)$ are symmetric. 
Then they remain symmetric for all $t\ge0$ thanks to the symmetric right-hand side of \eqref{eq:Z_i_X_i}.
Unlike \eqref{eq:candidate0}, a PI (proportional-integral) type coupling is employed in that $Z_i$ plays the role of integrator's state, and both $Z_i$ and $X_i$ are communicated to the neighbors.
The coupling gain is $\gamma_c > 0$, and the parameter $k_c > 0$ is introduced to adjust the speed of the algorithm.
Note that the number $N$ does not appear in \eqref{eq:Z_i_X_i} unlike in \eqref{eq:candidate0}.
This change yields a shorter response time of the overall controller that will be designed in the next section, and the convergence is now, instead of $X_i(t) \to X_*$,
\begin{align}\label{eq:X_i_X^*/N}
	X_i(t) \rightarrow \frac{X_*}{N} \quad \text{as} \quad t \rightarrow \infty, \quad \forall i \in \cN.
\end{align}
The next theorem shows that \eqref{eq:X_i_X^*/N} is the case if $\gamma_c>0$ and $k_c>0$.
Moreoever, the convergence is exponential whose speed can be made arbitrarily fast by increasing both $\gamma_c$ and $k_c$.

\begin{thm}\label{thm:pi}
Define the error variable (matrix) $\eta_c$ as	
\begin{align}\label{eq:eta_c}
	\eta_c(t) := \begin{bmatrix}
		\col_{i=1}^N(X_i(t)-\frac{X_*}{N}) \\
		(R^T \otimes I_{n})\col_{i=1}^N(Z_i(t))-\tilde Z_*
	\end{bmatrix}		
\end{align}
where 
$$\tilde Z_* := (k_c/\gamma_c) ((\Lambda^+)^{-1}R^T\otimes I_{n}) \col_{i=1}^N(2B_iB_i^T).$$
If Assumptions (T1), (T2), and (T3) hold, and if $k_c >0$ and $\gamma_c >0$, then $\exists m_c >0$ and $\lambda_c >0$ s.t.
\begin{align}\label{eq:m_lambda_eta_c}
\|\eta_c(t)\| \le m_c e^{-\lambda_c t}\|\eta_c(0)\|,\quad\forall t \ge 0.
\end{align}
In addition, for given $\lambda_c > 0$, \eqref{eq:m_lambda_eta_c} holds if
\begin{align}\label{eq:DistBass_kgamma}
	 k_c \ge \lambda_{\max}(P)\lambda_c, \quad \gamma_c \ge \frac{6+\sqrt{4+\|P\|^2\|\bar A\|^2}}{2\lambda_2(\cL)\lambda_{\min}(P)}k_c
\end{align}
where
$P \in \bR^{n^2 \times n^2}$ is the positive definite matrix such that
\begin{align}\label{eq:Lyap_barA}
P\bar A + \bar A^T P=-2I_{n^2}
\end{align}
with $\bar A:=-((A+\beta I_n)\oplus(A+\beta I_n))$. 

\end{thm}

Proof of Theorem \ref{thm:pi} is given in Appendix \ref{app:2}.
The theorem specifies the convergence property of not only $X_i$ but also $Z_i$.

Now, we may set the feedback gain $F_i$ as
\begin{equation}\label{eq:F_i0}
F_i(t) = -\frac{1}{N} B_i^T (X_i(t))^{-1}
\end{equation}
which converges to $F_{*,i}$ in \eqref{eq:F_i_static} based on \eqref{eq:X_i_X^*/N}.
Unfortunately, during the transient period of the convergence, $X_i(t)$ may not even be invertible, so that \eqref{eq:F_i0} is not well-defined.
To avoid this problem, we introduce a filter whose operation resembles the zero-order hold with a sampling period $T$ as
\begin{align}
&\dot \Phi_i^X = 0_{n \times n}, & & \text{when$\!\!\!\!\mod(t,T) \not = 0$} \notag \\ 
&\Phi_i^X \leftarrow \begin{cases} \Phi_i^X, & \text{if $\det(X_i(t)) = 0$,} \\
X_i(t), & \text{if $\det(X_i(t)) \not = 0$,} \end{cases} & & \text{when$\!\!\!\!\mod(t,T) = 0$} \notag \\
&\underline X_i(t) := \Phi_i^X(t) \label{eq:Phi}
\end{align}
with any invertible initial condition $\Phi_i^X(0)$, where ${\rm mod}$ implies the modular operation.\footnote{In this paper, a discontinuous function is assumed to be right-continuous; for example, $\Phi_i^X(t) = \lim_{\epsilon \downarrow 0} \Phi_i^X(t+\epsilon)$, $\forall t$.}
Then, the state feedback gain $F_i$ is defined (instead of \eqref{eq:F_i0}) as
\begin{align}\label{eq:F_i1}
F_i(t) = -\frac{1}{N} B_i^T (\underline X_i(t))^{-1}, \qquad i \in \cN
\end{align}
which has the property that $\lim_{t\to\infty}F_i(t) = -B_i^T X_*^{-1} = F_{*,i}$ because $\underline X_i^{-1}$ exists for all time and $\lim_{t\to\infty} \underline X_i^{-1} = N X_*^{-1}$.

On the other hand, we can similarly obtain the injection gains $L_i$ asymptotically.
Indeed, the gain 
\begin{equation}\label{eq:Li}
L_{*,i} = -Y_*^{-1}C_i^T, \qquad i \in \cN,
\end{equation}
where $Y_*$ is the solution to the dual of \eqref{eq:Bass}:
\begin{align}\label{eq:Bass_dual}
	-(A^T+\beta I_n)Y-Y(A^T+\beta I_n)^T+2C^TC=0,
\end{align}
makes $(A+L_* C)$ Hurwitz.
Therefore, with
\begin{subequations}\label{eq:W_i_Y_i}
\begin{align}
&\dot W_i = -\gamma_o \sum_{j\in\cN_i} (Y_j-Y_i)  \\
&\dot Y_i = k_o[-(A^T+\beta I_n) Y_i - Y_i (A^T+\beta I_n)^T + 2C_i^TC_i] \notag \\
&\qquad + \gamma_o \sum_{j\in\cN_i} (Y_j-Y_i) + \gamma_o \sum_{j\in\cN_i} (W_j-W_i)
\end{align}
\end{subequations}
where $W_i(0)$ and $Y_i(0)$ are chosen to be symmetric, $\gamma_o>0$ and $k_o > 0$, we obtain the symmetric solution  and $Y_i(t) \to Y_*/N$ as $t \to \infty$ as given in Theorem \ref{thm:pi}. Therefore, by employing the following filter
\begin{align}
	&\dot \Phi_i^Y = 0_{n \times n}, \qquad\qquad\qquad\qquad\qquad~\text{when$\!\!\!\!\mod(t,T) \not = 0$} \notag \\ 
	&\Phi_i^Y \leftarrow \begin{cases} \Phi_i^Y, & \text{if $\det(Y_i(t)) = 0$,} \\
		Y_i(t), & \text{if $\det(Y_i(t)) \not = 0$,} \end{cases}  \quad\text{when$\!\!\!\!\mod(t,T) = 0$} \notag \\
	&\underline Y_i(t) := \Phi_i^Y(t)\label{eq:PhiY}
\end{align}
with any invertible initial condition $\Phi_i^Y(0)$, the injection gain $L_i$ is given by
\begin{equation}\label{eq:Li1}
L_i(t) = -\frac1N (\underline Y_i(t))^{-1} C_i^T,\qquad \forall i\in\mathcal N.
\end{equation}
In summary, if the distributed control of a fixed graph in which no joining/leaving agents are of interest, the proposed control agent that has the identical dynamics given by \eqref{eq:prop}, \eqref{eq:Z_i_X_i}, \eqref{eq:F_i1}, \eqref{eq:W_i_Y_i}, \eqref{eq:Li1}, \eqref{eq:Phi}, and \eqref{eq:PhiY} for $X_i$ and $Y_i$ would stabilize the plant under Assumptions (T1), (T2), and (T3).
This control features distributed operation (F1) and decentralized design (F2) under (T1), (T2), and (T3).

\subsection{Removing Assumption (T1): Distributed Computation of $N$}

The aforementioned algorithm needs the information of $N$, the number of active agents in the network, in \eqref{eq:prop}, \eqref{eq:F_i1}, and \eqref{eq:Li1}.
This number $N$ is the global information in terms of individual control agents, and it varies with time if some agents join or leave during the operation.
Therefore, in order to extend the aforementioned algorithm for the plug-and-play operation (F3), the number $N$ needs to be estimated in a distributed manner by each agent only through a local communication with their neighbors.
In fact, there are a few methods for this purpose in the literature, e.g., \cite{IShames2012,CBaquero2012} but we employ the idea of the network size estimator \cite{DLee2018} which is again based on the blended dynamics theorem.
This is because it does fit into our purpose of the plug-and-play operation.

In this subsection, we present a distributed computation of $N$ assuming that the graph is fixed and Assumption (T2) holds.
Suppose that each control agent includes the dynamics
\begin{align}\label{eq:d_zeta_i}
\begin{split}
\dot \psi_i &= -\gamma_s \sum_{j\in\mathcal N_i}(\nu_j-\nu_i) \\
\dot \nu_i &= k_s + \gamma_s \sum_{j\in\cN_i}(\nu_j-\nu_i) + \gamma_s \sum_{j\in\mathcal N_i}(\psi_j-\psi_i)
\end{split}
\end{align}
where $\psi_i \in \bR$, $\nu_i \in \bR$, $\gamma_s\in\bR$ is the coupling gain, and $k_s\in\bR$ is the scaling factor.
Moreover, suppose that an additional agent (called \emph{informer} and labeled as 0) is added to the network, whose dynamics is
\begin{align}\label{eq:d_zeta_0}
\begin{split}
\dot \psi_0 &= -\gamma_s \sum_{j\in\mathcal N_0}(\nu_j-\nu_0) \\
\dot \nu_0 &= - k_s \nu_0 + \gamma_s \sum_{j\in\cN_0}(\nu_j-\nu_0) + \gamma_s \sum_{j\in\mathcal N_0}(\psi_j-\psi_0)
\end{split}
\end{align}
where $\cN_0$ is the index set of the neighbors of the agent 0.
Since agent 0 is added with bidirectional communication to any agent/agents of the connected graph $\cG$, the overall graph $\bar \cG$ having $N+1$ agents is also connected.
In this case, the blended dynamics becomes
\begin{align}\label{eq:zeta_blended}
	\dot s = -\frac{k_s}{N+1}s + \frac{k_sN}{N+1}.
\end{align}
Since it is a stable dynamics whose solution $s(t)$ converges to $N$, we expect that
\begin{align*}
\nu_i(t) \rightarrow N \quad \text{as} \quad t \rightarrow \infty, \quad i=0,1,\cdots,N
\end{align*}
which is the case if $\gamma_s>0$ and $k_s>0$.
This is guaranteed in the following theorem, whose proof is in Appendix \ref{app:3}.

\begin{thm}\label{thm:DLee2018}
Define the error vector
\begin{align}\label{eq:eta_s}
\eta_s := \begin{bmatrix}
		\col_{i=0}^N(\nu_i(t)-N) \\
		\bar R^T \col_{i=0}^N(\psi_i(t))-\tilde\psi_*
	\end{bmatrix}		
\end{align}
where $\bar R\in\bR^{(N+1)\times N}$ is a matrix satisfies \eqref{eq:R} with respect to the graph Laplacian $\bar\cL$ of $\bar\cG$ and
\begin{align}
\tilde \psi_* := \frac{k_s}{\gamma_s} (\bar \Lambda^+)^{-1}\bar R^T\begin{bmatrix}-N\\ 1_N\end{bmatrix}.
\end{align}
Under Assumption (T2), the network of \eqref{eq:d_zeta_i} and \eqref{eq:d_zeta_0} with $k_s >0$ and $\gamma_s >0$ is exponentially stable in the sense that
\begin{align}\label{eq:m_lambda_eta_s}
|\eta_s(t)| \le m_s e^{-\lambda_s t}|\eta_s(0)|, \quad \forall t \ge 0
\end{align}
with $m_s > 0$ and $\lambda_s > 0$.
In addition, for given $\lambda_s > 0$, \eqref{eq:m_lambda_eta_s} holds if
\begin{align}\label{eq:DistN_kgamma}
	k_s\geq \frac{24N+30+2\sqrt 5}{(2-\sqrt 2)}\lambda_s,\quad \gamma_s> \frac{(N+1)k_s}{\lambda_2(\bar\cL)}.
\end{align}

\end{thm}

Even if $\nu_i(t)$ converges to $N$, its value may become zero during its transient.
Therefore, we introduce a static filter as
\begin{equation}\label{eq:nfilter}
\underline\nu_i(t) := \Phi_i^\nu(t) = \max\{\nu_i(t), 0.5 \}.
\end{equation}

\subsection{The Proposed Control Agent}

Putting all the findings so far together, the identical control agent is now presented as Algorithm \ref{algo}.

\begin{algorithm}
\caption{An agent for channel $i$ performs:}\label{algo}


\noindent\textbf{Preset:} $\gamma_c\!>\!0$, $k_c\!>\!0$, $\gamma_o\!>\!0$, $k_o\!>\!0$, $\gamma_s\!>\!0$, $k_s\!>\!0$, $T\!>\!0$

\noindent\textbf{Require:} Informer agent 0 always runs \eqref{eq:d_zeta_0} in the network.

\noindent\textit{When joins for channel $i$:} Set $(A, B_i, C_i)$, $\beta$, $\Phi_i^X$, $\Phi_i^Y$

\noindent\rule{\columnwidth}{0.5pt}

\noindent{\bf Input:} $y_i$

\noindent{\bf Dynamics:}
	\begin{itemize}
	\item \eqref{eq:Z_i_X_i} with \eqref{eq:Phi} yielding $\underline X_i(t)$
	\item \eqref{eq:W_i_Y_i} with \eqref{eq:PhiY} yielding $\underline Y_i(t)$
	\item \eqref{eq:d_zeta_i} with \eqref{eq:nfilter} yielding $\underline \nu_i(t)$
	\item a distributed state observer:
		\begin{align}
			\hspace{-0.5cm} \dot{\hat x}_i = (A + \underline \nu_i B_i \underline F_i) \hat x_i + \underline \nu_i \underline L_i (C_i \hat x_i - y_i) + \gamma_i \sum_{j \in \cN_i} (\hat x_j-\hat x_i) \label{eq:prop1}
		\end{align}
	where
	\begin{align}
			&\underline F_i(t)=-\frac{1}{\underline \nu_i}B_i^T\underline X_i^{-1}\label{eq:F_i(t)}\\
			&\underline L_i(t)=-\frac{1}{\underline \nu_i}\underline Y_i^{-1}C_i^T\label{eq:L_i(t)}\\
			&\gamma_i(t) = 1+\frac{\underline\nu_i^2}{4}\bigg(\theta_i+\theta_i^2\kappa_i+4\|\underline X_i^{-1}\|^2\kappa_i\sqrt{1+\theta_i^2\kappa_i^2}\,\bigg)\notag
	\end{align}
	with
	\begin{align}\label{eq:tk}\hspace{-2mm}\begin{split}
				\theta_i(t) &=\|A\|+\|\underline Y_i^{-1}\|+2\|\underline X_i^{-1}\|\\
				\kappa_i(t)
				&=\frac{1}{\beta}\|\diag(\frac{\underline X_i^{-1}}{\underline\nu_i},\underline\nu_i\underline Y_i)\|\|\diag(\underline\nu_i\underline X_i,\frac{\underline Y_i^{-1}}{\underline\nu_i})\|.
				\end{split}
	\end{align}
	\end{itemize}
\noindent{\bf Communicate:} $\nu_i$, $\psi_i$, $Z_i$, $X_i$, $W_i$, $Y_i$, $\hat x_i$

\noindent{\bf Output:} $u_i = \underline F_i \hat x_i$

\end{algorithm}

Compared to \eqref{eq:prop}, the dynamics \eqref{eq:prop1} shows that the global information such as $\gamma$ is replaced by its local correspondence.
While the preset parameters are embedded in all control agents when they are constructed, it is supposed that $B_i$ and $C_i$ are learned when any agent happens to be linked to a channel (say, the channel $i$ for convenience).
The parameter $A$ can be preset, or learned as $B_i$ and $C_i$. 
The parameter $\beta$ can be preset, or set when $A$ is learned if a rule to choose $\beta$ is preset; 
\begin{equation}\label{eq:preset}
\text{e.g.,} \quad \beta = 1 + \max\{0, - \min_{1\le j\le n}\{ {\rm Re}(\lambda_j(A)) \}\}
\end{equation}
so that $A+\beta I_n$ is Hurwitz.
Initial conditions for all dynamic equations are freely chosen except that $\Phi_i^X(0)$ and $\Phi_i^Y(0)$ for all $i$ should be any nonsingular matrix.
This is for the system \eqref{eq:Phi} to generate nonsingular matrices forever, and not a restriction.
If an agent has to leave the network (by intention or by accident), then it can leave abruptly without any particular handshaking procedure.

From now on, let us suppose that some agents join or leave the network during the control operation, and let $\{t_k : k \in \bN\}$ be the time sequence of joining/leaving events with $t_0 = 0$ being the initial time of operation.
Also, let $\cN(t)$ be the index set of active agents at time $t$, which is right-continuous, and let $|\cN(t)|$ be the cardinality of $\cN(t)$.
Now, instead of Assumptions (T2) and (T3), we assume:
\begin{assmpt}\label{ass:PnP_1}
For all $t \ge 0$,
\begin{enumerate}
\item the communication network $\cG(t)$ among the active agents is undirected, unweighted, and connected, and
\item the plant $(A,\row_{i \in \cN(t)}(B_i), \col_{i \in \cN(t)}(C_i))$ is controllable and observable.
\end{enumerate}
\end{assmpt}

Now suppose an event of joining/leaving occurs at some time $t_k$ and no more event occurs afterward.
Then, the following theorem shows asymptotic behavior of the closed-loop system.

\begin{thm}\label{thm:main}
Suppose that Assumption \ref{ass:PnP_1} holds, the control agents run Algorithm \ref{algo}, and the set $\cN(t)$ remains the same for all $t \ge t_k$, so that we define $\cN^k := \cN(t)$.
Then, the closed-loop consisting of the plant \eqref{eq:plant} and the distributed observer \eqref{eq:prop1} tends to the LTI system consisting of \eqref{eq:plant} and \eqref{eq:prop}.
Moreover, the time-varying closed-loop dynamics of 
\eqref{eq:plant} and \eqref{eq:prop1} 
is asymptotically stable, and thus, the plant's state $x$ converges to the origin.
\end{thm}

\begin{proof}
To prove asymptotic stability of the time-varying closed-loop, we employ \cite[Lemma 2.1]{Bai87}\footnote{The origin of a time-varying linear system $\dot x = A(t) x$ is asymptotically stable if $A(t)$ is bounded and $\lim_{t \to \infty} A(t)$ exists and is Hurwitz.}. 
To use this assertion, we note from Theorems 2 and 3 that, for all $i \in \cN^k$, the following limits hold: $\underline \nu_i(\infty) = |\cN^k|=:N$, $\underline F_i(\infty)=-B_i^TX_*^{-1}$, and $\underline L_i(\infty)=-Y_*^{-1}C_i^T$, where $X_*$ and $Y_*$ are the solutions of Bass' equations \eqref{eq:Bass} and \eqref{eq:Bass_dual} with $\row_{i \in \cN^k}(B_i)$ and $\col_{i \in \cN^k}(C_i)$ instead of $B$ and $C$, respectively.
Moreover, $\theta_i(\infty)$, $\kappa_i(\infty)$, and $\gamma_i(\infty)$ in Algorithm \ref{algo} are the same for all $i \in \cN^k$, respectively, as
\begin{align}
\theta_i(\infty) &= \|A\|+N\|Y_*^{-1}\|+2N\|X_*^{-1}\| =: \theta_* \notag\\
\kappa_i(\infty) &=\frac{\lambda_{\max}(\diag(X_*^{-1},Y_*))}{\beta\lambda_{\min}(\diag(X_*^{-1},Y_*))}=: \kappa_* \notag\\
\gamma_i(\infty) &= 1\!\!+\!\!\frac{N^2}{4}\bigg[\theta_*\!\!+\!\theta_*^2\kappa_*\!\!+N^2\left\|X_*^{-1}\right\|^2\kappa_*\!\sqrt{1+\theta_*^2\kappa_*^2}\,\bigg]\notag\\
&=:\gamma_*.\label{eq:gamma_infty}
\end{align}
Now, the closed-loop \eqref{eq:plant} and \eqref{eq:prop1} with limit parameters $\underline F_i(\infty)$, $\underline L_i(\infty)$, $\underline \nu_i(\infty)$, and $\gamma_i(\infty)$ is an LTI system, which is equal to \eqref{eq:plant} and \eqref{eq:prop} with the parameters $F_i$, $L_i$, and $\gamma$ replaced by \eqref{eq:F_i_static}, \eqref{eq:Li}, and $\gamma_*$, respectively, and is exponentially stable according to Theorem \ref{thm:separation}; it is seen that $\theta_* \ge \theta$ because of \eqref{eq:B_i_C_i}; $\kappa = \kappa_*$ with $M_1=X_*^{-1}$, $Q_1=\beta X_*^{-1}$, $M_2=Y_*$, and $Q_2=\beta Y_*$; and therefore, $\gamma_* > \bar\gamma$ with \eqref{eq:Mohar91} and \eqref{eq:B_i_C_i}.
Then, the asymptotic stability of the time-varying closed-loop \eqref{eq:plant} and \eqref{eq:prop1} follows by \cite[Lemma 2.1]{Bai87}.
\end{proof}

It is obvious that the proposed distributed control, i.e., the network of active control agents running Algorithm \ref{algo}, features (F1) distributed operation and (F2) decentralized design under Assumption \ref{ass:PnP_1}.
It also features (F3) plug-and-play in the sense that, whenever some agents join/leave, a new asymptotically stable closed-loop of \eqref{eq:plant} and \eqref{eq:prop1} is formed without manipulating the initial conditions of all active agents in the network.
A practical example is presented in the next section to demonstrate these features.

The closed-loop system can be regarded as a kind of switched system, in which the switching is triggered by the joining/leaving agents.
A transient is expected immediately after each switching, but the closed-loop remains stable as long as no more switching occurs (as shown by Theorem \ref{thm:main}).
This also implies that, as long as the switching does not occur too often, the closed-loop still remains stable.

It is well-known that, even if each mode of a switched system is exponentially stable, there may exist a switching signal which makes the system unstable \cite{Liberzon}.
This happens when the event that a new switching occurs during the transient of the previous switching, repeats endlessly.
A well-known remedy to this pathology is to restrict the switching not to occur too often such that $t_{k+1}-t_k \ge \tau_d > 0$ where $\tau_d$ is called a {\em dwell-time} \cite{Liberzon}.
In practice, it is helpful to shorten the transient period by increasing the gains $\gamma_c$, $k_c$, $\gamma_o$, $k_o$, $\gamma_s$, and $k_s$,
according to Theorem \ref{thm:pi} and Theorem \ref{thm:DLee2018}, so that $X_i(t)$, $Y_i(t)$, and $\nu_i(t)$ converge more quickly. 
On the other hand, taking larger $\beta$ also accelerates the convergence of $x(t)$ to the origin by Theorem~\ref{thm:Bass}.

Unlike the conventional switched systems, the proposed distributed control does not have a dwell-time unless more restrictions are imposed.
For example, since a newly joining agent can have arbitrarily large initial conditions, the transient period may become arbitrarily large so that any finite dwell-time becomes invalidated in view of stability.
Also, since the proposed control does not limit the number of active agents, the transient period may become very large if large number of new agents join at the same time.
In Section \ref{sec:dwell}, we compute a dwell-time for the proposed distributed control after imposing a few restrictions to avoid the pathological cases.

\section{An Illustrative Example}\label{sec:casestudy}

\begin{figure}[t]
	\centering
		\includegraphics[width=0.21\textwidth]{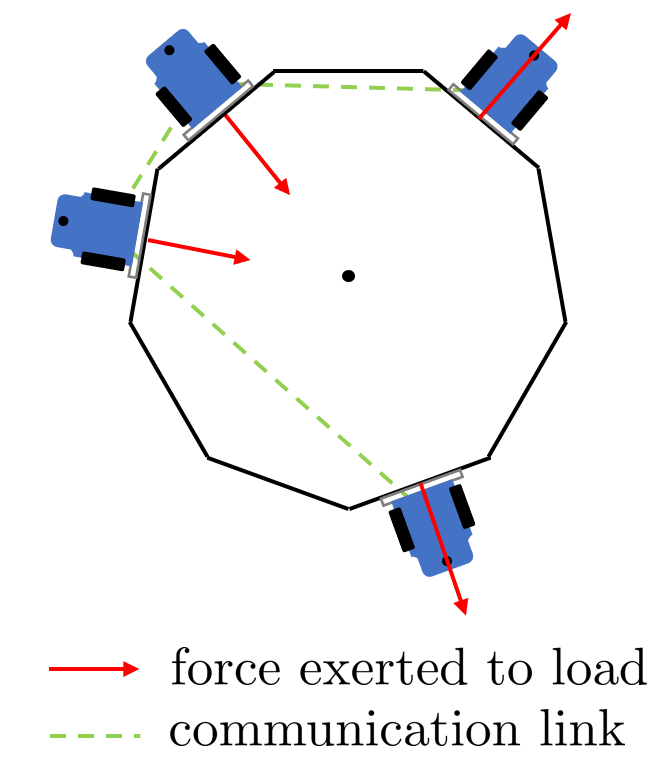}
		\includegraphics[width=0.24\textwidth]{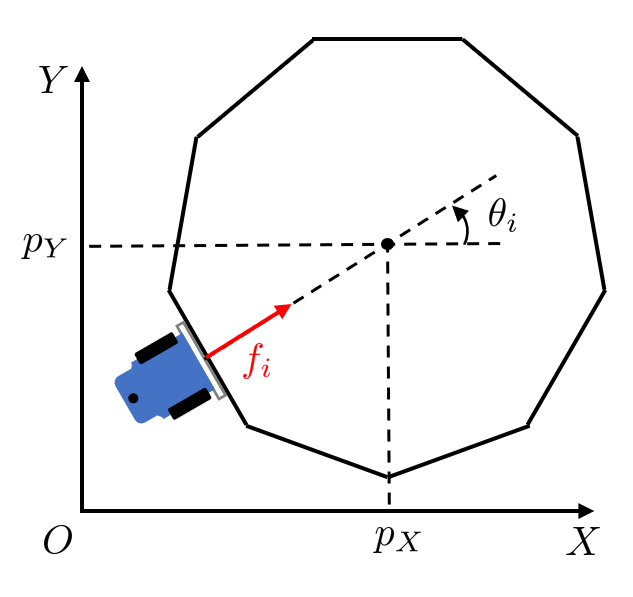}
	\caption{Agents performing cooperative transportation: Each agent occupies one edge of the load and exerts a push or pull force toward the center of load.}
	\label{fig:config}
\end{figure} 

To illustrate the utility of the proposed distributed control, let us consider the problem of cooperative load transportation by networked mobile robots as an example.
The load is assumed to be in a shape of a regular odd-number-sided polygon and, at each edge, one mobile robot can be attached (see Fig.~\ref{fig:config}). 
Each robot can exert pushing or pulling force $f_i$ in the normal direction, denoted as $[\cos\theta_i, \sin\theta_i]^T \in \bR^2$. 
We also restrict our attention to the case where each force $f_i$ is heading towards the center of load so that the load has only the translational motion without rotation, which simplifies the problem.
The goal is to transport the load to the desired location $p^d=[p_X^d,p_Y^d]^T\in\bR^2$ in a distributed fashion.
Let $p=[p_X,p_Y]^T\in\bR^2$ and $v\in\bR^2$ denote the position and the velocity of the load's center, respectively.
It is assumed that the load can measure the relative position $(p-p^d)$ and deliver this information to the attached robots who cannot measure their location.
No one can measure the velocity.
Then the dynamics of the load is of the form \eqref{eq:plant} with state $x:=[(p-p^d)^T,v^T]^T\in\bR^4$ and
\begin{align*}
	A&=\begin{bmatrix}
		0 & 0 & 1 & 0 \\
		0 & 0 & 0 & 1 \\
		0 & 0 & 0 & 0 \\
		0 & 0 & 0 & 0 \\
		\end{bmatrix},\;
	B_i=\frac{1}{M}\begin{bmatrix}
		0 \\ 0 \\ \cos\theta_i \\ \sin\theta_i
	\end{bmatrix},\;
	C_i=\begin{bmatrix}
		1 & 0 \\
		0 & 1 \\
		0 & 0 \\
		0 & 0
	\end{bmatrix}^T
\end{align*}
with $M\in\bR$ being the mass of the load.
The goal is achieved if we distributedly stabilize the origin of the multi-channel plant described above.

Note that none of individual triplet $(A,B_i,C_i)$ is controllable but any two pair of $B_i$'s and $C_i$'s ensure controllability and observability, so that Assumption \ref{ass:PnP_1} holds with at least two robots.
We assume that the informer is equipped in the load, and each robot can read $A$, $B_i$, and $C_i$ when the robot is attached at any edge.
The plant parameters are $M=1$ and $p^d=[100,150]^T$, and the preset parameters of Algorithm \ref{algo} are $k_s=\gamma_s=k_c=\gamma_c=k_o=\gamma_o=3$ and $\beta=0.7$.
The simulation scenario is depicted in Fig.~\ref{fig:scene}, and its results are in Fig.~\ref{fig:result1}, Fig.~\ref{fig:result2}, and Fig.~\ref{fig:result3}.
Another simulation has been performed with 100 control agents for 31-polygon as in Fig.~\ref{fig:video}.
In all simulations, we added measurement noise of intensity 1 when measuring $y_i$, and added random noise of magnitude $10^{-6}$ to all communicated quantities in order to take into account small delays and quantizations.


\begin{figure}[t]
	\centering
	\subfigure[]{
		\includegraphics[width=0.14\textwidth]{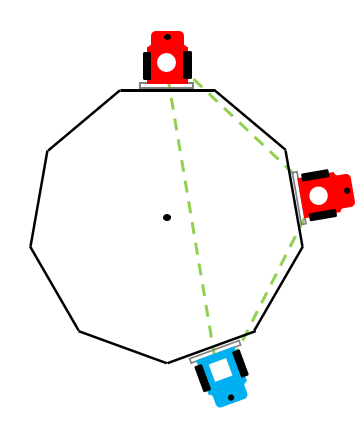}
	}
	\subfigure[]{
		\includegraphics[width=0.14\textwidth]{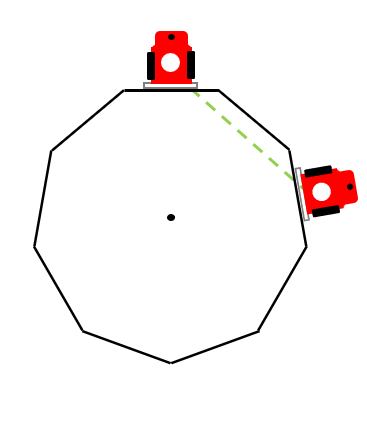}
	}
	\subfigure[]{
		\includegraphics[width=0.14\textwidth]{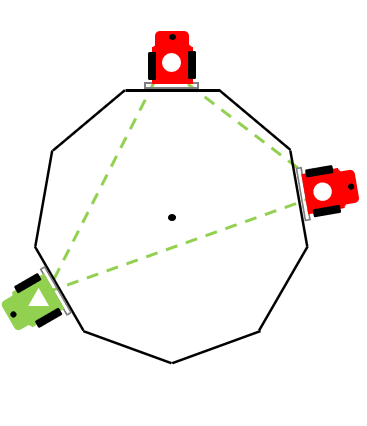}
	}
	\caption{The simulation starts with three agents as (a). At $t=3$, a blue agent leaves the network as (b). At $t=7$, a green agent joins as (c).}
	\label{fig:scene}
\end{figure}

\begin{figure}
	\centering
	\includegraphics[width=1\columnwidth]{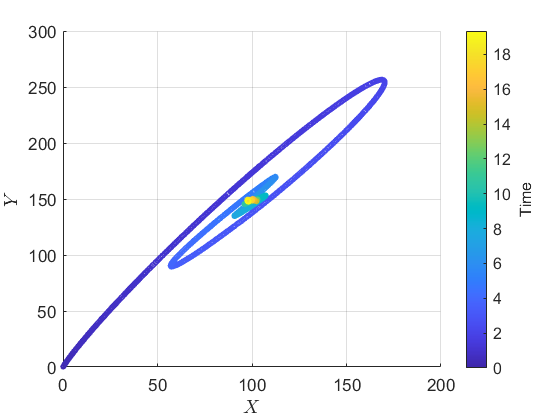}
	\caption{Trajectory of the position $p(t)$ of the load over time.}
	\label{fig:result1}
\end{figure}

\begin{figure}
	\centering
	\includegraphics[width=1\columnwidth]{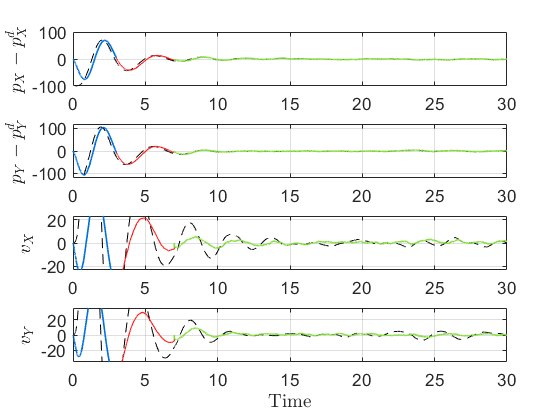}
	\caption{State $x(t)=[(p(t)-p^d)^T,v(t)^T]^T$ is drawn as a black dashed curve while one of its estimates $\hat x_i$ is drawn as colored solid curves.}
	\label{fig:result2}
\end{figure}

\begin{figure}
	\centering
	\includegraphics[width=1\columnwidth]{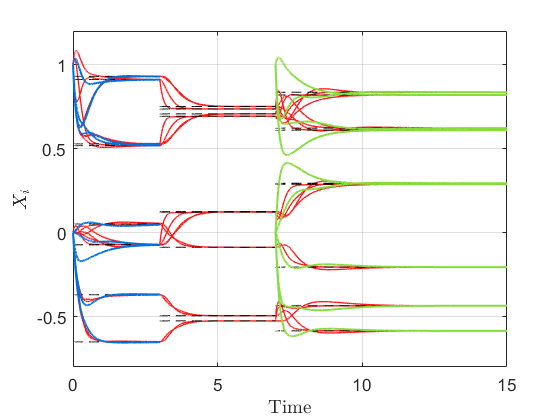}
	\caption{16 elements of $X_*/|\mathcal N(t)|$ are drawn as black dashed lines, and 16 elements of one $X_i(t)$ are drawn as colored solid curves.}
	\label{fig:result3}
\end{figure}

\begin{figure}
	\centering
	\includegraphics[width=.8\columnwidth]{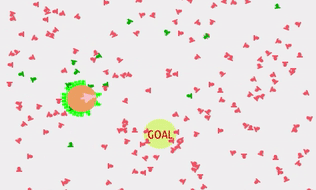}
	\small{\url{https://www.youtube.com/watch?v=ZIWL7xUXNVA}}
	\caption{Green agents are active, attached to the orange load, while the red agents are wandering around the load. Darkgreen agents are those who left the load with their battery empty. Animation is available at the above URL.} 
	\label{fig:video}
\end{figure}

\section{Dwell-time Analysis} 
\label{sec:dwell}


\subsection{Restrictions imposed for the existence of a dwell-time}

The closed-loop system consists of the plant state $x$ and the controller states $\hat x_i$, $Z_i$, $X_i$, $\Phi_i^X$, $W_i$, $Y_i$, $\Phi_i^Y$, $\psi_i$, and $\nu_i$, for all $i \in \cN(t)$.
And it is immediately seen that the sizes of the states vary depending on the time $t$.
Nevertheless, since $\cN(t) =: \cN^k$ remains the same during the $k$-th interval $[t_k,t_{k+1})$, we can define the error vector $e^k(t) := \col_{i \in \cN^k}(e_i(t))$ (where $e_i = \hat x_i - x$) that may have different size for each interval $[t_k,t_{k+1})$.
For the $k$-th interval, let us define $X_*^k$ and $Y_*^k$ as the solutions of Bass' equations \eqref{eq:Bass} and \eqref{eq:Bass_dual} with $\row_{i \in \cN^k}(B_i)$ and $\col_{i \in \cN^k}(C_i)$ instead of $B$ and $C$, respectively.
Then, letting  
$F_i^{k,\infty} := -B_i^T(X_*^k)^{-1}$, $L_i^{k,\infty} := -(Y_*^k)^{-1}C_i^T$, $\nu_i^{k,\infty} := |\cN^k|$, and $\gamma_*^{k,\infty}$ as the evaluation of $\gamma_*$ in \eqref{eq:gamma_infty} with the parameters $N$, $X_*$, and $Y_*$ replaced by $|\cN^k|$, $X_*^k$, $Y_*^k$ (where the superscript $\infty$ means its limit if we suppose the $k$-th time interval is $[t_k,\infty)$ as done in the proof of Theorem \ref{thm:main}), the dynamics of $x$ and $e^k$ for $[t_k,t_{k+1})$ can be obtained 
from \eqref{eq:plant} and \eqref{eq:prop1} as
\begin{align}
\begin{bmatrix}
 \dot x \\ \dot e^k
\end{bmatrix} &= \bigg( \underbrace{\begin{bmatrix}
			A+\sum_{i \in \cN^k} B_iF_i^{k,\infty} & \row_{i \in \cN^k}(B_i F_i^{k,\infty}) \\
			H_k & G_k -\gamma_*^{k,\infty}(\cL_k \otimes I_n)
 \end{bmatrix}}_{=: \; \Omega_k} \notag \\
&\qquad\qquad\qquad + \underbrace{\begin{bmatrix}
		\Delta_k^{xx}(t) & \Delta_k^{xe}(t) \\
		\Delta_k^{ex}(t) & \Delta_k^{ee}(t)
\end{bmatrix}}_{=: \; \Delta_k(t)} \bigg) \begin{bmatrix}
 x \\ e^k
 \end{bmatrix} \label{eq:dot_xe}
\end{align}
where 
\begin{align}\label{eq:GH}
\begin{split}
G_k &= \diag_{i \in \cN^k} (A + |\cN^k| L_i^{k,\infty} C_i + |\cN^k| B_i F_i^{k,\infty}) \\
 &\quad -1_{|\cN^k|} \otimes \row_{i \in \cN^k}(B_i F_i^{k,\infty}) \\
H_k &= \col_{i \in \cN^k}(|\cN^k|B_i F_i^{k,\infty} - \sum\nolimits_{i \in \cN^k} B_i F_i^{k,\infty}) \\
\Delta_k^{xx} &= \sum\nolimits_{i \in \cN^k} B_i (\underline F_i(t)-F_i^{k,\infty}) \\
\Delta_k^{xe} &= \row_{i \in \cN^k}(B_i(\underline F_i(t) - F_i^{k,\infty})) \\
\Delta_k^{ex} &= \col_{i \in \cN^k}(B_i(\underline \nu_i(t) \underline F_i(t)-\nu_i^{k,\infty} F_i^{k,\infty})) \\
&\quad -1_{|\cN^k|} \otimes \sum\nolimits_{i \in \cN^k} B_i (\underline F_i(t)-F_i^{k,\infty}) \\
\Delta_k^{ee} &= \diag_{i \in \cN^k}(B_i(\underline \nu_i(t) \underline F_i(t)-\nu_i^{k,\infty} F_i^{k,\infty})) \\
&\quad +\diag_{i \in \cN^k}((\underline \nu_i(t) \underline L_i(t)-\nu_i^{k,\infty} L_i^{k,\infty})C_i) \\
&\quad -1_{|\cN^k|} \otimes \row_{i \in \cN^k}(B_i(\underline F_i(t)-F_i^{k,\infty})) \\
&\quad -(\diag_{i \in \cN^k}(\gamma_i(t)-\gamma_*^{k,\infty}) \cdot \cL_k) \otimes I_n .
\end{split}
\end{align}

This $(x,e^k)$-dynamics is the key equation, and the role of other dynamics \eqref{eq:Z_i_X_i}, \eqref{eq:Phi}, \eqref{eq:W_i_Y_i}, \eqref{eq:PhiY}, \eqref{eq:d_zeta_i}, and \eqref{eq:d_zeta_0} for $Z_i$, $X_i$, $\Phi_i^X$, $W_i$, $Y_i$, $\Phi_i^Y$, $\psi_i$, and $\nu_i$ are to make $\underline F_i(t)$, $\underline L_i(t)$, $\underline \nu_i(t)$, and $\gamma_i(t)$ converge to $F_i^{k,\infty}$, $L_i^{k,\infty}$, $\nu_i^{k,\infty}$, and $\gamma_*^{k,\infty}$ respectively, so that 
\begin{equation}\label{eq:delta0}
\Delta_k(t) \to 0.
\end{equation}
Then, based on the fact that $\Omega_k$ is Hurwitz by Theorem \ref{thm:separation}, our goal is to have a dwell-time $\tau_d>0$ such that, if $t_{k+1}-t_k \ge \tau_d$, $\forall k$, then, with 
\begin{equation}\label{eq:V}
V(t) := |x(t)|^2 + \sum_{i \in \cN(t)} |\hat x_i(t)-x(t)|^2 = |x(t)|^2 + |e^k(t)|^2
\end{equation}
where $k$ is such that $t \in [t_k,t_{k+1})$, it holds that
\begin{align}\label{eq:exp_stab_V}
V(t) \le m e^{-\lambda t} V(0), \quad \forall t \ge 0
\end{align} 
for some $m>0$ and $\lambda>0$.
The function $V(t)$ is not continuous due to the switchings at $t_k$, but can still be used for proving the convergence of $x$ to the origin, as in Fig.~\ref{fig:StabilityConcept}.

There are three major difficulties to obtain the uniform convergence of \eqref{eq:exp_stab_V} for the system \eqref{eq:dot_xe}.
The first one is that $\Omega_k$ is just a Hurwitz matrix without any restriction, so that its convergence rate may become arbitrarily slow and it is hopeless to expect \eqref{eq:exp_stab_V}.
To avoid this difficulty, we impose the following assumption, which makes $\Omega_k$ an element of a finite set.

\begin{assmpt}\label{ass:barN}
Suppose that the system \eqref{eq:plant} has a finite number of input-output channels; that is, there exists a constant $\bar N\in\mathbb N$ such that $\cN(t)\subset \{1,\cdots, \bar N\}$ for all $t$. Let $\cB := \{B_1,\dots,B_{\bar N}\}$ and $\cC := \{C_1,\dots,C_{\bar N}\}$. Let $\Sigma$ be the collection of $\sigma \subset \{1,\dots,\bar N\}$ such that $(A,\row_{i \in \sigma}(B_i), \col_{i \in \sigma}(C_i))$ is controllable and observable, and let $\bG_\sigma$ be the collection of all undirected, connected, and unweighted graphs of the nodes $\sigma$.
Define $\bF := \{ (A,\row_{i \in \sigma}(B_i), \col_{i \in \sigma}(C_i)) \times \bG_\sigma : \sigma \in \Sigma \}$.
Assume $\bF$ is non-empty.
\end{assmpt}

Then, it is clear that $\bF$ is a finite set, and for each element of $\bF$, there are\footnote{We suppose $\beta$ of \eqref{eq:Bass} is predetermined such that $-(A+\beta I_n)$ be Hurwitz. An example is \eqref{eq:preset}.} the corresponding Bass' solutions $X_*$ and $Y_*$ (see \eqref{eq:Bass} and \eqref{eq:Bass_dual}), and let $\bX$ and $\bY$ be the collection of all such solutions $X_*$ and $Y_*$, respectively.
For convenience, let us denote $\|\bX^{-1}\| := \max\{\|X_*^{-1}\| : X_*\in\bX\}$ and $\|\bX\| := \max\{\|X_*\| : X_*\in\bX\}$. 
Define $\|\bY^{-1}\|$ and $\|\bY\|$ similarly.

The second difficulty is that, even though $\Delta_k(t)$ tends to zero as \eqref{eq:delta0} so that \eqref{eq:dot_xe} asymptotically becomes a stable LTI system, it is hopeless to expect \eqref{eq:exp_stab_V} if $\Delta_k(t)$ becomes arbitrarily large during the transient before it gets small.
In fact, this phenomenon does happen if $\det(X_i(t))$ becomes arbitrarily close to zero (see \eqref{eq:Phi}) during the transient of $X_i(t)$.
To prevent this phenomenon, we now modify the update rule of \eqref{eq:Phi} as
\begin{align}\label{eq:Phi2}
\begin{split}
\Phi_i^X &\leftarrow \begin{cases} \Phi_i^X, & \text{if } \; \lambda_{\min}(X_i(t)) < \frac{1}{2\bar N\|\bX^{-1}\|}, \\
X_i(t), & \text{if } \; \lambda_{\min}(X_i(t)) \ge \frac{1}{2\bar N\|\bX^{-1}\|}, \end{cases}
\end{split}
\end{align}
whenever $\!\!\mod(t,T) = 0$, with any initial $\Phi_i^X(0)$ such that $\lambda_{\min}(\Phi_i^X(0)) \ge 1/2\bar N\|\bX^{-1}\|$.
The same modification is needed for $\Phi_i^Y$ in \eqref{eq:PhiY}, but no modification is necessary for the filter \eqref{eq:nfilter} because its lower bound is already uniform over $\bF$.
As a result, we obtain the uniform bound for $\underline F_i$ as
\begin{align}\label{eq:Fibar}
\|\underline F_i(t)\| \le \frac{1}{|\underline\nu_i(t)|} \|B_i^T\| \|\underline X_i(t)^{-1}\| \le 4 \bar N \|\bX^{-1}\|, \quad \forall t,
\end{align}
which follows from \eqref{eq:F_i(t)}, \eqref{eq:nfilter}, \eqref{eq:B_i_C_i}, and \eqref{eq:Phi2}.
Similarly, we have $\|\underline L_i(t)\| \le 4 \bar N \|\bY^{-1}\|$, $\forall t$.
To implement \eqref{eq:Phi2}, however, both $\|\bX^{-1}\|$ and $\|\bY^{-1}\|$, and $\bar N$ need to be precomputed.

The third difficulty for uniform convergence \eqref{eq:exp_stab_V} of the system \eqref{eq:dot_xe} is that, when a new agent join at time $t_k$, its states may be set arbitrarily large (as briefly discussed at the end of Section \ref{sec:2}).
For example, if $\hat x_i(t_k)$ is set arbitrarily large, then
it is hopeless to expect the uniform value $m$ in \eqref{eq:exp_stab_V}.
Similarly, if $X_i(t_k)$ can be set arbitrarily large, then it may take arbitrarily long until $X_i(t)$ gets reasonably close to $X_*/|\cN^k|$ so that the control agent starts stabilizing action, which may prevent the existence of a dwell-time.
A solution to this issue is to install a policy for setting the initial condition of newly joining agents as follows.
For this, we recall the fact that, when a new agent joins the network at time $t_k$, it can communicate with at least one neighbor that belongs to either $\cN(t_k) \cap \cN(t_k^-)$\footnote{For a time-varying quantity $f(t$), we denote $f(t_k^-)$ as the left-hand limit of $f$ at $t_k$, i.e., $f(t_k^-):=\lim_{\ep\uparrow0}f(t_k+\ep)$.} (implying the neighbor has been active before the switching) or $\cN(t_k) \setminus \cN(t_k^-)$ (implying the neighbor is another newly joining agent).
The following Algorithm is appended to Algorithm \ref{algo}, in which $a_i$ is a flag to indicate its state is not arbitrary.

\begin{figure}
	\centering
	\includegraphics[width=.7\columnwidth]{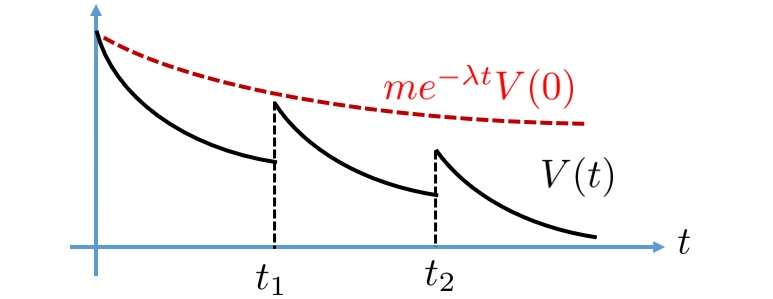}
	\caption{A typical behavior of $V(t)$ when control agents are allowed to join/leave the network.}
	\label{fig:StabilityConcept}
\end{figure}

\begin{algorithm}
\caption{When an agent joins for channel $i$ at $t_k$:}\label{algo2}
\begin{algorithmic}[1]
\renewcommand{\algorithmicrequire}{Let $\pi_i := (\hat x_i, Z_i, X_i, W_i, Y_i, \psi_i, \nu_i)$ and  $a_i=$ {\footnotesize$\begin{cases}T,~\text{if}~t_k=0\\F,~\text{otherwise.}\end{cases}$}}
\REQUIRE 
\renewcommand{\algorithmicensure}
\ENSURE
\WHILE{$a_i=F$}
\FOR{$j\in\mathcal N_i(t_k)$}
\IF{$a_j=T$}
    \STATE copy $\pi_i \leftarrow \pi_j$
    \STATE $a_i=T$
    \STATE {\bf break}
\ENDIF
\ENDFOR
\ENDWHILE
\STATE set $\Phi_i^X \leftarrow (1/2\bar N\|\bX^{-1}\|) I_n$ and $\Phi_i^Y \leftarrow (1/2\bar N\|\bY^{-1}\|) I_n$
\end{algorithmic} 
\end{algorithm}

Note that this algorithm does not violate the plug-and-play feature (F3) because it does not require resetting the initial condition for all other agents in the network.
Moreover, compared to the plug-and-play policy studied in \cite{Riverso2013,Bendtsen2013,Zeilinger2013}, the above policy is mild in the sense that it does not involve manipulation on the neighboring agents (and no manipulation is required when they leave the network).
Lastly, we note that, without such a copy process in Algorithm \ref{algo2}, we cannot expect the asymptotic convergence of $x(t)$ to zero because of the persistent perturbation caused by new initial conditions (even if they are uniformly bounded) of joining agents.
This is in sharp contrast to the stability analysis of switched systems in which the state is continuous in time under switchings.

Now, we are ready to state the main result of this section.

\begin{thm}\label{thm:dwell-time}
Under Assumptions \ref{ass:PnP_1} and \ref{ass:barN}, for any initial configuration of the plant and active control agents, there is a dwell-time $\tau_d > 0$ such that the closed-loop consisting of the plant \eqref{eq:plant} and the networked active control agents running Algorithms \ref{algo} and \ref{algo2} with the update rules of \eqref{eq:Phi} and \eqref{eq:PhiY} replaced like \eqref{eq:Phi2}, satisfy \eqref{eq:exp_stab_V} with positive numbers $m$ and $\lambda$, as long as $t_{k+1}-t_k \ge t_d$ for all $k$.
\end{thm}

The next subsection proves Theorem \ref{thm:dwell-time}, which also shows how to obtain $\tau_d$ (in \eqref{eq:tau_d}).

\subsection{Proof of Theorem \ref{thm:dwell-time}}


We first look at the propagation of state values across the switching.
For this, let $\pi_i$ in Algorithm \ref{algo2} be a column vector (by vectorizing $X_i$ for example), i.e., $\pi_i \in \bR^q$ where $q = n + 4n^2 + 2$, and let $\pi^k(t) := \col_{i \in \cN^k}(\pi_i(t))$.
Then, there exists a sequence of matrix $\{S_k\}$ obtained by Algorithm \ref{algo2} such that
\begin{align}\label{eq:S_k}
\pi^k(t_k) = (S_k \otimes I_q) \pi^{k-1}(t_k^-).
\end{align}
Note that $S_k \in \bR^{|\cN^k| \times |\cN^{k-1}|}$ whose rows have only one element of 1 and all others are 0.
Therefore, it holds that
\begin{align}\label{eq:property_Sk}
S_k 1_{|\cN^{k-1}|} = 1_{|\cN^k|}, \quad \|S_k\| \le \sqrt{|\cN^k|}.
\end{align}
Moreover, let us define the average $\bar\pi^k \in \bR^q$ and the rest $\tilde\pi^k \in \bR^{(|\cN^k|-1)q}$ as
\begin{align}\label{eq:def_bar_Tilde_eta}
\bar\pi^k := (\frac{1}{|\cN^k|} 1_{|\cN^k|}^T \otimes I_q) \pi^k, \quad
\tilde\pi^k := (R_k^T \otimes I_q) \pi^k
\end{align}
where $R_k$ is the matrix satisfying \eqref{eq:R} with respect to the graph Laplacian $\cL_k$, so that we have the inverse of \eqref{eq:def_bar_Tilde_eta} as $\pi^k = 1_{|\cN^k|} \bar\pi^k + R_k \tilde\pi^k$.

\smallskip\noindent{\bf Claim 1:} It holds that
\begin{align}
\bar\pi^k(t_k) &= \bar\pi^{k-1}(t_k^-) + (\frac{1}{|\cN^k|} 1_{|\cN^k|}^T S_k R_{k-1} \otimes I_q) \tilde\pi^{k-1}(t_k^-) \notag \\
\tilde\pi^k(t_k) &= (R_k^T S_k R_{k-1} \otimes I_q) \tilde\pi^{k-1}(t_k^-). \label{eq:bar_tilde_eta}
\end{align}

\begin{proof} The above follows from
\begin{align*}
\bar\pi^k(t_k)
		&= (\frac{1}{|\cN^k|} 1_{|\cN^k|}^T \otimes I_q)\pi^k(t_k) \\
		&= (\frac{1}{|\cN^k|} 1_{|\cN^k|}^T \otimes I_q)(S_k \otimes I_q)\pi^{k-1}(t_k^-) \\
		&= (\frac{1}{|\cN^k|} 1_{|\cN^k|}^T S_k 1_{|\cN^{k-1}|} \otimes I_q) \bar\pi^{k-1}(t_k^-) \\
		&\quad + (\frac{1}{|\cN^k|} 1_{|\cN^k|}^T S_k R_{k-1} \otimes I_q) \tilde\pi^{k-1}(t_k^-) \\
\tilde\pi^k(t_k) &= (R_k^T \otimes I_q)\pi^k(t_k) \\
&= (R_k^T S_k 1_{|\cN^{k-1}|} \otimes I_q) \bar\pi^{k-1}(t_k^-) \\
&\quad + (R_k^T S_k R_{k-1} \otimes I_q) \tilde\pi^{k-1}(t_k^-)
\end{align*}
in which $R_k^T S_k 1_{|\cN^{k-1}|} = R_k^T 1_{|\cN^k|}=0$ by \eqref{eq:R}.
\end{proof}

Now, we look at how $X_i(t)$ in Theorem \ref{thm:pi} is uniformly bounded with respect to $X_*^k/|\cN^k|$. 
For this, with $X^k := \col_{i \in \cN^k}(X_i)$, $\tilde Z^k := (R_k^T \otimes I_n) \col_{i \in \cN^k}(Z_i)$, and $\tilde Z_*^k := (k_c/\gamma_c)((\Lambda_k^+)^{-1} R_k^T \otimes I_{n}) \col_{i \in \cN^k} (2B_iB_i^T)$ (where $\Lambda_k^+$ corresponds to $\cL_k$ according to \eqref{eq:R}), define 
\begin{equation}\label{eq:etac}
\eta_c^k := \begin{bmatrix} X^k - (1_{|\cN^k|} \otimes I_n) \frac{X_*^k}{|\cN^k|} \\ \tilde Z^k - \tilde Z_*^k \end{bmatrix}.
\end{equation}
Note that the above definition is in fact the same as \eqref{eq:eta_c}.
Let $\bar m_c := \max_{\bF}\{m_c\}$ and $\underline \lambda_c := \min_{\bF}\{\lambda_c\}$, in which the values of $m_c$ and $\lambda_c$ are obtained from Theorem \ref{thm:pi} for each element of $\bF$.

\smallskip\noindent{\bf Claim 2:} If $t_{k+1}-t_k \ge \tau_c$ for all $k \ge 0$ where
\begin{align}\label{eq:tau_c}
\tau_c := \frac{1}{\underline\lambda_c}\ln(2 \bar m_c \bar N \|\bX^{-1}\| \max\{|\eta_c^0(0)|, r_c\})
\end{align}
then 
\begin{align}
\|\eta_c^0(t)\| &\le \bar m_c e^{-\underline\lambda_c t} \|\eta_c^0(0)\|, & & t \in [0,t_1), \label{eq:exp_bd_eta_c1} \\
\|\eta_c^k(t)\| &\le \bar m_c e^{-\underline\lambda_c(t-t_k)} r_c, & & t \in [t_k,t_{k+1}), \; k \ge 1 \label{eq:exp_bd_eta_c2}
\end{align}
where
\begin{align}\label{eq:r_c}
r_c := \frac{1}{2\sqrt{\bar N} \|\bX^{-1}\|} + \frac{k_c \bar N^3}{\gamma_c} + 2\sqrt{\bar N}\|\bX\|.
\end{align}

\begin{proof}
Since $\bar m_c$ is the maximum of all possible $m_c$ and $\underline \lambda_c$ is the minimum of all possible $\lambda_c$, \eqref{eq:exp_bd_eta_c1} holds itself by Theorem \ref{thm:pi}.
In order to prove \eqref{eq:exp_bd_eta_c2}, we first claim that
\begin{align}\label{eq:eta_c_t_k_t_k-}
\|\eta_c^k(t_k)\| \le \sqrt{\bar N}\|\eta_c^{k-1}(t_k^-)\| + \frac{k_c \bar N^3}{\gamma_c} + 2\sqrt{\bar N}\|\bX\|.
\end{align}
Indeed, it follows from \eqref{eq:S_k} and \eqref{eq:property_Sk} that
\begin{align*}
&X^k(t_k) - (1_{|\cN^k|} \otimes I_n) \frac{X_*^k}{|\cN^k|} \\
&= X^k(t_k) - (1_{|\cN^k|} \otimes I_n) \Big( \frac{X_*^{k-1}}{|\cN^{k-1}|} + \frac{X_*^k}{|\cN^k|} - \frac{X_*^{k-1}}{|\cN^{k-1}|} \Big) \\
&= (S_k \otimes I_n) \left( X^{k-1}(t_k^-) - (1_{|\cN^{k-1}|} \otimes I_n) \frac{X_*^{k-1}}{|\cN^{k-1}|} \right) \\
&\qquad - (1_{|\cN^k|} \otimes I_n) \left( \frac{X_*^k}{|\cN^k|} - \frac{X_*^{k-1}}{|\cN^{k-1}|} \right)
\end{align*}
and by \eqref{eq:bar_tilde_eta} it is obtained that
\begin{align*}
\tilde Z^k(t_k) - \tilde Z_*^k &= (R_k^T S_k R_{k-1} \otimes I_n)[\tilde Z^{k-1}(t_k^-)-\tilde Z_*^{k-1}] \notag\\
&+ [(R_k^T S_k R_{k-1} \otimes I_n) \tilde Z_*^{k-1} - \tilde Z_*^k].
\end{align*}
By combining above two equalities, it is obtained that
\begin{align*}
\eta_c^k(t_k)
&= \begin{bmatrix}
	S_k\otimes I_n & 0 \\ 0 & R_k^T S_k R_{k-1} \otimes I_n
	\end{bmatrix} \eta_c^{k-1}(t_k^-) \notag \\
&+ \begin{bmatrix}
	-(1_{|\cN^k|} \otimes I_n)(\frac{X_*^k}{|\cN^k|}-\frac{X_*^{k-1}}{|\cN^{k-1}|}) \\
	(R_k^T S_k R_{k-1} \otimes I_n) \tilde Z_*^{k-1} - \tilde Z_*^k
\end{bmatrix}
\end{align*}
which yields, by \eqref{eq:property_Sk} and \eqref{eq:R}, that
\begin{align}\label{eq:bd_eta_tk}
\|\eta_c^k(t_k)\| &\le \sqrt{|\cN^k|} \big[ \|\eta_c^{k-1}(t_k^-)\| +\|X_*^{k-1}\|+\|X_*^k\| \notag\\
&\quad + \|\tilde Z_*^{k-1}\| + \|\tilde Z_*^k\| \big].
\end{align}
Here, it follows from the definition of $\tilde Z_*^k$, the fact that $\|(\Lambda^+_k)^{-1}\| \le 1/\lambda_2(\cL_k) \le |\cN^k|^2/4$, and \eqref{eq:B_i_C_i}, that
\begin{align}\label{eq:bd_tilde_Z*}
\|\tilde Z_*^k(t)\| \le \frac{\bar N^2 k_c}{4\gamma_c}\|\col_{i\in\cN_k}(2B_iB_i^T)\| \le \frac{k_c \bar N^{\frac{5}{2}}}{2\gamma_c}.
\end{align}
By applying \eqref{eq:bd_tilde_Z*} and $\|X_*\|\le \|\bX\|$ to \eqref{eq:bd_eta_tk}, we obtain \eqref{eq:eta_c_t_k_t_k-}.
Then, \eqref{eq:exp_bd_eta_c2} can be shown based on \eqref{eq:tau_c}, \eqref{eq:exp_bd_eta_c1}, \eqref{eq:r_c}, and \eqref{eq:eta_c_t_k_t_k-} because, at every $t = t_k^-$ for $k \ge 1$, $\|\eta_c^{k-1}(t_k^-)\| \le 1/2\bar N \|\bX^{-1}\|$.
Indeed, since $t_1-t_0 \ge \tau_c$, it holds that $|\eta_c^0(t_1^-)| \le 1/2\bar N\|\bX^{-1}\|$, which, combined with \eqref{eq:eta_c_t_k_t_k-}, implies that $|\eta_c(t_1)|\le r_c$. 
By repeating this process for $k=2,\cdots$, we obtain \eqref{eq:exp_bd_eta_c2}.
\end{proof}

Noting that $\|X_i(t)-X_*^k/|\cN^k|\| \le \|\eta_c^k(t)\|$, $i \in \cN^k$, if $\tau_{k+1}-\tau_k \ge \tau_c$, $\forall k \ge 0$, then, a uniform bound is obtained as
\begin{align}
\left\|X_i(t)-\frac{X_*^k}{|\cN^k|}\right\| &\le \bar m_c e^{-\underline\lambda_c(t-t_k)} \max\{|\eta_c^0(0)|,r_c\} \label{eq:exp_bd_err_Xi} \\
&\le \bar m_c \max\{|\eta_c^0(0)|,r_c\}, \label{eq:unif_bd_eta_c}
\end{align}
for $i \in \cN_k$ and $t \in [t_k, t_{k+1})$, $\forall k \ge 0$.

Similar inequalities as \eqref{eq:exp_bd_err_Xi} and \eqref{eq:unif_bd_eta_c} 
can be obtained for $Y_i$, with similarly defined $\tau_o$.

In addition, similar analysis can be performed for $\nu_i$.
Define, from \eqref{eq:eta_s} like done as in \eqref{eq:etac}, 
$$\eta_s^k := \begin{bmatrix} \col_{i \in \cN^k}(\nu_i - |\cN^k|) \\ \bar R_k^T \col_{i \in \cN^k}(\psi_i) - \frac{k_s}{\gamma_s}(\bar\Lambda_k^+)^{-1} \bar R_k^T \begin{bmatrix}-|\cN^k| \\ 1_{|\cN^k|}\end{bmatrix} \end{bmatrix}.$$
And, let $\bar m_s:=\max_{\bF}\{m_s\}$ and $\underline\lambda_s:=\min_{\bF}\{\lambda_s\}$, obtained from Theorem \ref{thm:DLee2018}, and define
\begin{align}
	r_s &:= \frac{\sqrt {\bar N}}{2} +\bar N\sqrt{\bar N+1} + \frac{k_s}{2\gamma_s}(\bar N+1)^{\frac{7}{2}} \\
	\tau_s &:= \frac{1}{\underline\lambda_s}\ln\Big(2\bar m_s\max\{|\eta_s^0(0)|, r_s\}\Big). \label{eq:r_s}
\end{align}
Then, it is seen that, if $\tau_{k+1}-\tau_k \ge \tau_s$, $\forall k \ge 0$, a uniform bound is obtained as
\begin{align}
|\nu_i(t)-|\cN^k|| &\le \bar m_s e^{-\underline\lambda_s(t-t_k)} \max\{|\eta_s^0(0)|,r_s\} \label{eq:exp_bd_err_zetai} \\
	&\le \bar m_s\max\{|\eta_s^0(0)|,r_s\}, \label{eq:unif_bd_eta_s}
\end{align}
for $i \in \cN^k$ and $t \in [t_k, t_{k+1})$, $\forall k \ge 0$.

\medskip

Now look at \eqref{eq:dot_xe} again and recall that the Hurwitz matrix $\Omega_k$ corresponds to an element of the finite collection $\bF$.
Hence, if we define $\underline \lambda := \min_{P \in \bF_P} \lambda_{\min}(P)$ and $\bar \lambda := \max_{P \in \bF_P} \lambda_{\max}(P)$ where $\bF_P := \{ P > 0 : P\Omega + \Omega^T P = -I_n, \Omega$ corresponds to an element in $\bF \}$.

\smallskip\noindent{\bf Claim 3:} 
If $t_{k+1}-t_k \ge \tau_1 := \max\{\tau_c,\tau_o,\tau_s\}$, then there exists a uniform bound $D_0$ such that
\begin{equation}\label{eq:D0}
	\|\Delta_k(t)\| \le D_0, \quad \forall t \in [t_k,t_{k+1}), \; \forall k \ge 0.
\end{equation}
Moreover, there exists $\tau_2 \ge 0$ such that, if $t_{k+1}-t_k \ge \tau_1 + \tau_2$,
\begin{align}\label{eq:Delta_bd_tau_2}
\|\Delta_k(t)\| \le \frac{\underline\lambda}{4\bar\lambda^2}, \quad \forall t \in [t_k+\tau_1+\tau_2, t_{k+1}), \; \forall k \ge 0.
\end{align}

\begin{proof}
Referring to the definitions of $\Delta_k^{xx}$, $\Delta_k^{xe}$, $\Delta_k^{ex}$, and $\Delta_k^{ee}$ in \eqref{eq:GH}, where $B_i$, $C_i$, $|\cN^k|$ and $\|\cL_k\|$ are uniformly bounded (see \eqref{eq:Anderson85}) under Assumption \ref{ass:barN}, it holds that
\begin{align}
	\|\Delta_k\| &\le 4\bar N^{\frac{3}{2}} \max_{i\in\cN^k} \|\underline{F}_i-F_i^{k,\infty}\|
	+\bar N\max_{i\in\cN^k}\left\|\underline{Y}_i^{-1}\!\!\!\!-(\underline Y_*^k)^{-1}\right\|\notag\\
	&\quad+2\bar N\max_{i\in\cN^k}\left\|\underline{X}_i^{-1}\!\!\!\!-(\underline X_*^k)^{-1}\right\| +\bar N \max_{i\in\cN^k}|\gamma_i-\gamma_*^{k,\infty}|.\label{eq:ubd_Dk}
\end{align}
where $\underline X_*^k:=X_*^k/|\cN^k|$ and $\underline Y_*^k:=Y_*^k/|\cN^k|$.
Note that $\Delta_k(t)$ is uniformly bounded, if $\gamma_i(t)$ is uniformly bounded, because $F_i^{k,\infty}$, $\underline X_*^k$, $\underline Y_*^k$, and $\gamma_*^{k,\infty}$ are uniformly bounded under Assumption \ref{ass:barN} and so are $\underline F_i(t)$, $\underline X_i(t)^{-1}$, and $\underline Y_i(t)^{-1}$ based on \eqref{eq:Fibar} and \eqref{eq:Phi2}. According to \eqref{eq:tk}, the uniform boundedness of $\gamma_i(t)$ requires additionally that $\underline X_i(t)$, $\underline Y_i(t)$, and $\underline \nu_i(t)$ are uniformly bounded, which is true by \eqref{eq:unif_bd_eta_c} and \eqref{eq:unif_bd_eta_s} and justifies the claim \eqref{eq:D0}.

To show \eqref{eq:Delta_bd_tau_2}, we note that both $\|\underline{F}_i-F_i^{k,\infty}\|$ and $|\gamma_i-\gamma_*^{k,\infty}|$ in \eqref{eq:ubd_Dk} are uniformly bounded by $D_1(\|\underline{X}_i-\underline X_*^k\|+\|\underline{Y}_i-\underline Y_*^k\|+\|\underline{X}_i^{-1}\!\!\!\!-(\underline X_*^k)^{-1}\|+\|\underline{Y}_i^{-1}\!\!\!\!-(\underline Y_*^k)^{-1}\|+|\underline \nu_i-\nu_i^{k,\infty}|)$ because those are Lipschitz continuous functions and the variables belong to the uniformly bounded set.
By applying these bounds to \eqref{eq:ubd_Dk}, it is seen that there is $\tau_2$ such that \eqref{eq:Delta_bd_tau_2} holds because $\underline{X}_i-\underline X_*^k$, $\underline{Y}_i-\underline Y_*^k$, $\underline{X}_i^{-1}\!\!\!\!-(\underline X_*^k)^{-1}$, $\underline{Y}_i^{-1}\!\!\!\!-(\underline Y_*^k)^{-1}$, and $\underline \nu_i-\nu_i^{k,\infty}$ converges to zero exponentially whose rate is uniform as shown in \eqref{eq:exp_bd_err_Xi} and \eqref{eq:exp_bd_err_zetai}
\footnote{Indeed, by design, for $t\in[t_k+\tau_c,t_{k+1})$, it can be shown that $X_i(t)^{-1}$ exists and that $\|\underline{X}_i^{-1}\!\!\!\!-(\underline X_*^k)^{-1}\|$ is uniformly bounded by $2\bar N^2\|\bX^{-1}\|^2\|\underline{X}_i-\underline X_*^k\|$. Similar bound holds for $\|\underline{Y}_i^{-1}\!\!\!\!-(\underline Y_*^k)^{-1}\|$ on $t\in[t_k+\tau_o,t_{k+1})$.}.
\end{proof}

Finally, let us prove \eqref{eq:exp_stab_V}.
First, with $W_k := \row(x^T,(e^k)^T) P_k \col(x,e^k)$ where $P_k$ is the positive definite solution to $P_k\Omega_k + \Omega_k^T P_k = -I_n$, the time derivative of $W_k$ along \eqref{eq:dot_xe} yields that
\begin{align}
\dot W_k 
&\le -\frac{W_k}{\bar\lambda}\left(1-\frac{2\bar\lambda^2\|\Delta_k(t)\|}{\underline\lambda}\right) . \label{eq:dW_2} 
\end{align}
Therefore, it follows that 
\begin{align}\label{eq:Delta_dW}
\dot W_k \le -\frac{1}{2\bar\lambda} W_k \quad\text{ if }\|\Delta_k(t)\| \le \frac{\underline\lambda}{4\bar\lambda^2}.
\end{align} 
By applying \eqref{eq:Delta_bd_tau_2} to \eqref{eq:Delta_dW} and combining \eqref{eq:D0} and \eqref{eq:dW_2}, we obtain
\begin{align*}
\dot W_k \le \left\{\begin{array}{ll}
		\left(-\frac{1}{2\bar\lambda}+2\frac{\bar\lambda}{\underline\lambda}D_0\right) W_k, & t \in[t_k, t_k+\tau_1+\tau_2) \\
		-\frac{1}{2\bar\lambda} W_k, & t\in[t_k+\tau_1+\tau_2,t_{k+1})
	\end{array}\right. .
\end{align*}
Then, noting that $\underline\lambda V(t) \le W_k(t) \le \bar\lambda V(t)$ for all $t \ge 0$ from \eqref{eq:V}, we obtain that
\begin{align}
V(t) &\le \frac{\bar\lambda}{\underline\lambda}e^{2\frac{\bar\lambda}{\underline\lambda}D_0(\tau_1+\tau_2)-\frac{1}{2\bar\lambda}(t-t_k)}V(t_k) \notag \\
&= e^{a_p-\lambda_p(t-t_k)} V(t_k),\qquad t\in[t_k,t_{k+1}),	\label{eq:Vt_Vt_k}
\end{align}
where 
\begin{align}
a_p := \ln \frac{\bar\lambda}{\underline\lambda} + 2\frac{\bar\lambda}{\underline\lambda}D_0(\tau_1+\tau_2),\quad
	\lambda_p:=\frac{1}{2\bar\lambda}.
\end{align}
On the other hand, referring to \eqref{eq:S_k} we can derive that
\begin{align}\begin{split}\label{eq:ratio_V_k}
V(t_k) &= \left\|\begin{bmatrix} x(t_k) \\ e^k(t_k) \end{bmatrix}\right\|^2
= \left\|\begin{bmatrix} x(t_k^-) \\ (S_k \otimes I_n) e^{k-1}(t_k^-) \end{bmatrix}\right\|^2 \\
&\le \left\|\begin{bmatrix} I_n & 0 \\ 0 & S_k \otimes I_n \end{bmatrix}\right\|^2 V(t_k^-) \\
&= \max\{1,\|S_k\|^2\} V(t_k^-) \\
&\le |\cN^k| V(t_k^-) \le \bar N V(t_k^-) .
\end{split}\end{align}
Now, set the dwell time $\tau_d$ and the constant $\lambda>0$ such that
\begin{align}\label{eq:tau_d}
\tau_d > \frac{\ln\bar N+a_p}{\lambda_p},\quad \lambda:=\lambda_p-\frac{\ln\bar N+a_p}{\tau_d}
\end{align}
which results in that
\begin{align}\label{eq:prop_dwell}
0<\lambda\le\lambda_p, \quad e^{(\ln\bar N+a_p-\lambda_p\tau_d)}\le e^{-\lambda\tau_d}.
\end{align}
Then by using \eqref{eq:prop_dwell}, we have 
\begin{align}
&e^{(\ln\bar N+a_p-\lambda_p\delta)} = e^{(\ln\bar N+a_p-\lambda_p\tau_d)}e^{-\lambda_p(\delta-\tau_d)} \notag \\
&\qquad \le e^{-\lambda\tau_d}e^{-\lambda_p(\delta-\tau_d)} \notag \\
&\qquad \le e^{-\lambda\tau_d}e^{-\lambda(\delta-\tau_d)} \le e^{-\lambda \delta}, \quad \forall\delta\ge\tau_d. \label{eq:prop_lambda}
\end{align}
Finally, for any $t$ such that $t_{i+1} > t \ge t_i$, it holds that
\begin{align}
V(t) &\le e^{a_p-\lambda_p(t-t_i)}V(t_i) \notag\\
	&\le e^{a_p-\lambda_p(t-t_i)} \bar N V(t_i^-) \notag\\
	&\le e^{a_p-\lambda_p(t-t_i)} \left(\bar N e^{a_p-\lambda_p(t_i-t_{i-1})}\right) V(t_{i-1}) \notag\\	
	&\vdots \notag\\
	&\le e^{a_p-\lambda_p(t-t_i)}\prod_{q=1}^i \left(e^{\ln{\bar N}+a_p-\lambda_p(t_q-t_{q-1})}\right)V(t_0) \notag\\
	&\le e^{a_p-\lambda(t-t_i)}\left(e^{-\lambda(t_i-t_0)}\right) V(0) \notag\\
	&\le e^{a_p-\lambda t} V(0) 
\end{align}
where \eqref{eq:prop_lambda} is used in the fifth inequality because $t_q-t_{q-1} \ge \tau_d$, with the inequality $\lambda<\lambda_p$.
Finally, the condition \eqref{eq:exp_stab_V} follows if we set
\begin{align}
m := e^{a_p}.
\end{align}
This completes the proof.
The technique of choosing the dwell time $\tau_d$ in this proof is originated from \cite[Lemma 2]{Morse1996}.

\section{Conclusion}
In this paper, we propose a distributed output-feedback control for networked agents to stabilize a linear multi-channel plant by using local interaction with the plant and local inter-agent communication.
The key features of our distributed scheme are its decentralized design and plug-and-play capability, which enables each agent to self-organize its own controller using only locally available information and allows any agents to join or leave the network without requiring any redesign of remaining agents while maintaining the closed-loop stability as long as the basic condition is met.
To enable plug-and-play stabilization, we have developed a distributed Bass' algorithm based on the intuition from the blended dynamics theorem.
In the future, we will extend the scope of this approach to other areas of distributed control regarding output regulation or optimal control.

\appendix

\subsection{Continued Proof of Theorem \ref{thm:separation}}\label{app:1}

Now, consider a Lyapunov function candidate $V$ such that $V(x,\bar e,\tilde e):=(1/2)(x^T M_1x+\bar\phi\bar e^T M_2 \bar e + \tilde\phi\tilde e^T\tilde e)$ where $\bar\phi>0$ and $\tilde\phi>0$ will be chosen later.
The time derivative of $V$ along \eqref{eq:dot_e_dot_x} becomes
\begin{align}
\dot V 
&= - x^TQ_1x + x^TM_1BF\bar e + x^T M_1 \square_1\tilde e \notag\\
&\quad -\bar\phi\bar e^T Q_2\bar e + \bar\phi\bar e^T M_2 \square_2\tilde e +\tilde\phi\tilde e^T \square_3 x +\tilde\phi\tilde e^T \square_4 \bar e \notag\\
&\quad   -\tilde\phi\tilde e^T[\gamma (\Lambda^+\otimes I_n) -\square_5]\tilde e \notag\\
&\le -\frac{\ep}{4}|x|^2 - \left(\bar\phi\ep-\frac{\delta^2\|BF\|^2}{\ep} \right)|\bar e|^2 \notag\\
&\quad + \left(\bar\phi\delta\|\square_2\|+\tilde\phi\|\square_4\|\right)|\bar e||\tilde e| \notag\\
&\quad -\tilde\phi\left(\gamma\lambda_2(\cL)-\|\square_5\|-\frac{\delta^2\|\square_1\|^2}{\ep\tilde\phi} -\frac{\tilde\phi\|\square_3\|^2}{\ep}\right)|\tilde e|^2 \label{eq:sep_dV_1}
\end{align}
where $\delta:=\lambda_{\max}(\diag(M_1,M_2))$ and $\ep:=\lambda_{\min}(\diag(Q_1,Q_2))$.
Note that, based on \eqref{eq:B_i_C_i}, \eqref{eq:R}, and $\|G\| \le \theta$, the following inequalities hold:
\begin{align}\begin{split}\label{eq:bd_square15}
	\|\square_1\|^2 &\le N\max_{i\in\cN}\|F_i\|^2,\quad \|\square_3\|^2 \le 4N^3 \max_{i\in\cN} \|F_i\|^2 \\
	\|\square_2\| &\le \frac{\theta}{\sqrt N},\quad \|\square_4\| \le \sqrt{N}\theta,\quad \|\square_5\|\le \theta.
\end{split}\end{align}
By using $\|BF\|\le N\max_{i\in\cN}\|B_iF_i\|$ with \eqref{eq:B_i_C_i} and by substituting \eqref{eq:bd_square15} into \eqref{eq:sep_dV_1}, it is obtained that
\begin{align}\label{eq:dV_THM_1}
\dot V &\le \!-\frac{\ep}{4}|x|^2\!\! -\!\!\left[\bar\phi\ep-\frac{\delta^2N^2\max\limits_{i\in\cN}\|F_i\|^2}{\ep}\right]\!\!\!|\bar e|^2 
\!\!+\!\theta\frac{\bar\phi\delta+\tilde\phi N}{\sqrt{N}}|\bar e||\tilde e| \notag\\
&\quad-\tilde\phi\left[\gamma\lambda_2(\cL)-\theta-\frac{N\max\limits_{i\in\cN}\|F_i\|^2}{\ep}\left(\frac{\delta^2}{\tilde\phi}+4N^2\tilde\phi\right)\right]|\tilde e|^2 .
\end{align}
Choose $\bar\phi := 2\delta^2N^2\max_{i\in\cN}\|F_i\|^2/\ep^2 + \tilde\phi N/\delta$.
Then the right-hand side of \eqref{eq:dV_THM_1} becomes negative definite if
\begin{align*}
	\gamma\lambda_2(\cL) &> \theta+\frac{\theta^2\delta}{\ep}+\frac{N\max\limits_{i\in\cN}\|F_i\|^2}{\ep}\left(\frac{\delta^2}{\tilde\phi}+4N^2\tilde\phi+\frac{\theta^2\delta^4}{\ep^2\tilde\phi}\right).
\end{align*}
Now choose $\tilde\phi := (\delta/2N) \sqrt{1 + \theta^2\delta^2/\ep^2}$.
Then, \eqref{eq:gamma_general_1} follows from the above condition.

\subsection{Proof of Theorem \ref{thm:pi}}\label{app:2}

Using the fact that ${\rm vec}(ABC) = (C^T \otimes A) {\rm vec}(B)$ for any matrices $A$, $B$, and $C$, the vectorized version of \eqref{eq:Z_i_X_i} through $\zeta_i:={\rm vec}(Z_i)$ and $\chi_i:={\rm vec}(X_i)$ is given by
\begin{align}\begin{split}\label{eq:nu_i_chi_i}
\dot\zeta_i &= -\gamma_c\sum_{j\in\cN_i}(\chi_j-\chi_i) \\
\dot\chi_i &= k_c[\bar A\chi_i+{\rm vec}(2B_iB_i^T)]+\gamma_c\sum_{j\in\cN_i}(\chi_j-\chi_i) \\ &\qquad+\gamma_c\sum_{j\in\cN_i} (\zeta_j-\zeta_i),\qquad i=1,2,\cdots,N.	
\end{split}\end{align}
Here, since a Kronecker sum of Hurwitz matrices is Hurwitz \cite{Neudecker69}, the matrix $\bar A$ is Hurwitz and the solution $P$ of \eqref{eq:Lyap_barA} is positive definite.
The whole system with states $\zeta:=\col_{i=1}^N(\zeta_i)$ and $\chi:=\col_{i=1}^N(\chi_i)\in\bR^{Nn^2}$ now reads as
\begin{align}
	\dot\zeta &= \gamma_c(\cL\otimes I_{n^2})\chi \label{eq:nu_chi} \\
	\dot\chi &= k_c[(I_N \otimes\bar A)\chi + W]-\gamma_c(\cL\otimes I_{n^2})\chi-\gamma_c(\cL\otimes
	I_{n^2})\zeta \notag
\end{align}
where $W:=\col_{i=1}^N({\rm vec}(2B_iB_i^T))$.
Let us consider the following coordinate change:
\begin{align}\begin{split}\label{eq:bar_tilde_nu_chi}
\begin{bmatrix}
\bar\chi \\ \tilde\chi
\end{bmatrix}
:=
\begin{bmatrix}
\frac{1}{N}1_N^T \otimes I_{n^2} \\ R^T \otimes I_{n^2}
\end{bmatrix}
\chi,\quad
\begin{bmatrix}
\bar\zeta \\ \tilde\zeta
\end{bmatrix}
:=		
\begin{bmatrix}
\frac{1}{N}1_N^T \otimes I_{n^2} \\ R^T \otimes I_{n^2}
\end{bmatrix}\zeta
\end{split}\end{align}
where $R\in\bR^{N \times (N-1)}$ is given in \eqref{eq:R}.
Here, the inverses are given by $\chi=(1_N\otimes I_{n^2})\bar\chi+(R\otimes I_{n^2})\tilde\chi$ and $\zeta=(1_N\otimes I_{n^2})\bar\zeta+(R\otimes I_{n^2})\tilde\zeta$.
Then \eqref{eq:nu_chi} through \eqref{eq:bar_tilde_nu_chi} becomes
\begin{align}\begin{split}\label{eq:nu_chi_2}
\dot{\bar \zeta} &= 0,\qquad \dot{\tilde \zeta} =\gamma_c (\Lambda^+\otimes I_{n^2})\tilde\chi \\
\dot{\bar\chi} &= k_c[\bar A\bar\chi +\frac{1}{N}\sum_{i=1}^N {\rm vec}(2B_iB_i^T)]\\
\dot{\tilde\chi} &= -\big[\gamma_c(\Lambda^+\otimes I_{n^2})-k_c(I_{(N-1)}\otimes\bar A)\big]\tilde\chi \\
&\quad -\gamma_c(\Lambda^+\otimes I_{n^2})\tilde\zeta+ k_c(R^T\otimes I_{n^2}) W
\end{split}\end{align}
where $\Lambda^+$ is defined in \eqref{eq:R}.
From \eqref{eq:bar_tilde_nu_chi} and \eqref{eq:nu_chi_2}, it is clear that $\bar\zeta(t) = (1/N) \sum_{i=1}^N \zeta_i(0)$ for all $t \ge 0$, and that $\bar \zeta$ does not affect the behavior of $\tilde \zeta$, $\bar \chi$, and $\tilde \chi$.

We claim that $\tilde \zeta(t) \to \tilde \zeta_*$, $\bar \chi(t) \to \bar \chi_*$, and $\tilde \chi \to 0$ as time tends to infinity with $k_c>0$ and $\gamma_c>0$, where
\begin{align}\label{eq:bar_chi*} 
\begin{split}
\tilde\zeta_* &= \frac{k_c}{\gamma_c}[(\Lambda^+)^{-1}\otimes I_{n^2}](R^T\otimes I_{n^2})W \\
\bar\chi_* &= -\bar A^{-1}\left(\frac{1}{N}\sum_{i=1}^N {\rm vec}(2B_iB_i^T)\right) .
\end{split}
\end{align}
This claim proves \eqref{eq:X_i_X^*/N} because the above equation implies $\bar A (N\bar \chi_*) + {\rm vec}(2BB^T) = 0$ 
while the vectorized form of \eqref{eq:Bass} is $\bar A {\rm vec}(X_*) + {\rm vec}(2BB^T) = 0$.

To show the claim, let $\eta := \bar\chi-\bar\chi_*$ and $\xi := \tilde\zeta-\tilde\zeta_*$.
Then, we have 
\begin{align}
\dot{\eta} &= k_c\bar A\eta \notag \\
\dot{\xi} &=\gamma_c (\Lambda^+\otimes I_{n^2})\tilde\chi \label{eq:eta_xi_tildechi} \\
\dot{\tilde\chi} &= -\big[\gamma_c(\Lambda^+\otimes I_{n^2})-k_c(I_{(N-1)}\otimes\bar A)\big]\tilde\chi -\gamma_c(\Lambda^+\otimes I_{n^2})\xi. \notag
\end{align}
Consider the Lyapunov function candidate $V=(1/2)\eta^TP\eta+\tilde V$ where
\begin{align*}
\tilde V:=\frac{1}{2}\begin{bmatrix}
\xi \\ \tilde\chi
\end{bmatrix}^T
\!\!\!\left(
\begin{bmatrix}
(\phi+2) &1\\ 1 &\phi+1
\end{bmatrix}
\otimes I_{(N-1)}
\otimes P
\right)
\begin{bmatrix}
\xi \\ \tilde\chi
\end{bmatrix}
\end{align*}
with $\phi>0$ being determined soon.
The time derivative of $V$ is given by
\begin{align}
&\dot V =-k_c|\eta|^2-\gamma_c\xi^T(\Lambda^+\otimes P)\xi +k_c\xi^T(I_{(N-1)}\otimes P\bar A)\tilde\chi \notag\\
&\qquad -\gamma_c\phi\tilde\chi^T(\Lambda^+\otimes P)\tilde\chi -k_c(\phi+1)|\tilde\chi|^2 \notag\\
&\le  -k_c|\eta|^2 -\gamma_c\lambda_2\lambda_{\min}(P)|\xi|^2+k_c\|P\|\|\bar A\||\xi||\tilde\chi|\notag\\
&\quad-(\gamma_c\lambda_2\lambda_{\min}(P)\phi+k_c\phi+k_c)|\tilde\chi|^2\label{eq:dV_eta_chi_xi}.
\end{align}
Thus, for given $\gamma_c>0$ and $k_c>0$, one can find $\phi>0$ satisfying
\begin{align}\label{eq:phiubound}
	\phi&>\frac{1}{(1+\frac{\gamma_c}{k_c}\lambda_2\lambda_{\min}(P))}\left( \frac{\|P\|^2\|\bar A\|^2}{4\frac{\gamma_c}{k_c}\lambda_2\lambda_{\min}(P)}-1 \right)
\end{align}
which renders $\dot V$ in \eqref{eq:dV_eta_chi_xi} to be negative definite.
By using $\|\eta_c\|_F^2=|\chi-(1_N\otimes I_{n^2})\bar\chi_*|^2+|\tilde \zeta-\tilde \zeta_*|^2=N|\eta|^2+|\xi|^2+|\tilde \chi|^2$ and $\|\eta_c\|\le \|\eta_c\|_F \le \sqrt{n}\|\eta_c\|$, we obtain \eqref{eq:m_lambda_eta_c} from the stability of \eqref{eq:eta_xi_tildechi}.

On the other hand, to guarantee the convergence rate $\lambda_c$, let us now consider $V$ with $\phi=1$.
Then, it can be seen that 
\begin{align}\label{eq:bd_tilde_V}
\tilde V \le \frac{5+\sqrt{5}}{4}\lambda_{\max}(P)(|\xi|^2+|\tilde\chi|^2).
\end{align}
Now, with $\phi = 1$, \eqref{eq:dV_eta_chi_xi} becomes
\begin{align*}
&\dot V \le -k_c |\eta|^2 \\
&-\frac{k_c}{2}\left(2\frac{\gamma_c}{k_c}\lambda_2\lambda_{\min}(P)+2-\sqrt{4+\|P\|^2\|\bar A\|^2}\right)(|\xi|^2+|\tilde\chi|^2) \\
&\le -\frac{2k_c}{\lambda_{\max}(P)}\frac{\eta^TP\eta}{2} \\ &\quad-\frac{2k_c\left(2\frac{\gamma_c}{k_c}\lambda_2\lambda_{\min}(P)+2-\sqrt{4+\|P\|^2\|\bar A\|^2}\right)}{\lambda_{\max}(P)(5+\sqrt{5})}\tilde V.
\end{align*}
Finally, the condition \eqref{eq:DistBass_kgamma} yields $\dot V \le -2\lambda_c V$.

\subsection{Proof of Theorem \ref{thm:DLee2018}}\label{app:3}

For clarity, let us denote $\bar R\in\bR^{N\times(N+1)}$ and $\bar \Lambda^+\in\bR^{N \times N}$ matrices satisfying \eqref{eq:R} with $\bar\cL$.
Let us define $\nu:=\text{col}_{i=0}^N(\nu_i)$ and $\psi:=\text{col}_{i=0}^N(\psi_i)$.
Then the dynamics \eqref{eq:d_zeta_i} and \eqref{eq:d_zeta_0} through the coordinate change $\bar \nu := 1_{N+1}^T\nu/(N+1)$, $\tilde \nu:=\bar R^T\nu$, $\bar \psi := 1_{N+1}^T\psi/(N+1)$, and $\tilde \psi:=\bar R^T\psi$, read as
\begin{align*}
	\dot{\bar\psi} &=0,\qquad\qquad \dot{\tilde\psi} =\gamma_s\bar \Lambda^+ \tilde\nu \\
	\dot{\bar\nu} &= -\frac{1_{N+1}^T}{N+1}(k_s J+\gamma_s\bar\cL)(1_{N+1}\bar\nu+\bar R\tilde\nu)+\frac{k_s N}{N+1} \\
	\dot{\tilde\nu} &= -\bar R^T(k_s J+\gamma_s\bar\cL)(1_{N+1}\bar\nu+\bar R\tilde\nu)  -\gamma_s\bar \Lambda^+\tilde\psi +k_s\bar R^Tw
\end{align*}
where $J=\diag(1,0,\cdots,0)\in\bR^{(N+1)\times(N+1)}$ and $w:=[0~1_N^T]^T$.
Since $\bar \psi(t) = \sum_{i=0}^N \psi_i(0)/(N+1)$ for all $t \ge 0$ and $\bar \psi$ does not affect other dynamics, we just analyze the dynamics of $\tilde \psi$, $\bar \nu$, and $\tilde \nu$, which has the equilibrium point $(\bar\nu_*,\tilde\nu_*,\tilde\psi_*)$ given by $\bar\nu_*=N$, $\tilde\nu_*=0$, and $\tilde \psi_*=(k_s/\gamma_s)(\bar \Lambda^+)^{-1}\bar R^T(w-NJ1_{N+1})$.
Define $\eta:=\sqrt{N+1}(\bar \nu-\bar \nu_*)$ and $\xi:=\tilde \psi-\tilde \psi_*$.
Then we have 
\begin{align}
	\begin{split}\label{eq:xi_eta_zeta}
		\dot \xi &= \gamma_s\bar \Lambda^+\tilde \nu,\\
		\dot \eta &= -\frac{1_{N+1}^T}{\sqrt {N+1}}(k_sJ+\gamma_s\bar {\mathcal L})\bigg(\frac{1_{N+1}}{\sqrt{N+1}}\eta+\bar R\tilde \nu\bigg)\\
		\dot {\tilde \nu} &= -\bar R^T(k_sJ+\gamma_s\bar {\mathcal L})\bigg(\frac{1_{N+1}}{\sqrt{N+1}}\eta+\bar R\tilde \nu\bigg)-\gamma_s\bar \Lambda^+\xi.
	\end{split}
\end{align}
Consider the Lyapunov function candidate $V=\phi|\eta|^2/2 + (\phi+1)|\xi|^2/2 + \phi|\tilde \nu|^2/2 + |\xi+\tilde \nu|^2/2$ 
where $\phi>0$ will be chosen later. The time derivative of $V$ along \eqref{eq:xi_eta_zeta} is given by
\begin{align}
	\dot V& = -k_s\phi[\eta^T~\,\tilde \nu^T]\begin{bmatrix}\frac{1_{N+1}^T}{\sqrt {N+1}}\\\bar R^T\end{bmatrix}\Big(J+\frac{\gamma_s}{k_s}\bar {\mathcal L}\Big)\begin{bmatrix}\frac{1_{N+1}}{\sqrt {N+1}}~\bar R\end{bmatrix}\begin{bmatrix}\eta\\\tilde \nu\end{bmatrix}\notag\\
	&\quad -\gamma_s\xi^T\bar \Lambda^+\xi-k_s\xi^T\bar R^TJ\frac{1_{N+1}}{\sqrt{N+1}}\eta-k_s\xi^T\bar R^TJ\bar R\tilde \nu \notag\\
	&\quad -k_s\tilde \nu^T\bar R^TJ\frac{1_{N+1}}{\sqrt{N+1}}\eta-k_s\tilde \nu^T\bar R^TJ\bar R\tilde \nu\notag\\
	&\leq -k_s\phi \lambda_{\min}\Big(J+\frac{\gamma_s}{k_s}\bar {\mathcal L}\Big)\Big(|\eta|^2+|\tilde \nu|^2\Big)-\gamma_s \lambda_2(\bar\cL)|\xi|^2\notag\\
	&\quad +k_s|\xi|\,|\eta|+k_s|\xi|\,|\tilde \nu|+k_s|\tilde \nu|\,|\eta|.\label{eq:dV_eta_xi_zeta}
\end{align}
where the last inequality holds since \cite[Lemma 1]{DLee2018} guarantees that the matrix $(J+(\gamma_s/k_s)\bar {\mathcal L})$ is positive definite with $\gamma_s>0$ and $k_s>0$. 
Thus, one can find $\phi>0$ satisfying
\begin{align*}
	\phi > \frac{1}{\lambda_{\min}(J+\frac{\gamma_s}{k_s}\bar {\mathcal L})}\Big(\frac{k_s}{2\gamma_s\lambda_2(\bar\cL)}+\frac{1}{2}\Big)
\end{align*}
which renders $\dot V$ of \eqref{eq:dV_eta_xi_zeta} to be negative definite. By using $|\eta_s|^2=|\nu-1_{N+1}\bar \nu_*|^2+|\tilde \psi-\tilde \psi_*|^2=|\eta|^2+|\xi|^2+|\tilde \nu|^2$, we obtain \eqref{eq:m_lambda_eta_s} from the stability of \eqref{eq:xi_eta_zeta}.

On the other hands, to guarantee the convergence rate $\lambda_s$, we note that 
\begin{align}\label{eq:ubd_V}
	V \leq \frac{2\phi+3+\sqrt 5}{2}(|\eta|^2+|\xi|^2+|\tilde \nu|^2).
\end{align}
In addition, under the second condition of \eqref{eq:DistN_kgamma}, it follows from \cite[Lemma 1]{DLee2018} that $\lambda_{\min}(J + (\gamma_s/k_s)\bar {\mathcal L}) \ge 1/(4(N+1))$.
Now, let $\phi=6(N+1)$.
Then, we have $\phi\lambda_{\min}(J+ (\gamma_s/k_s) \bar {\mathcal L}) \ge 3/2$.
By applying this lower bound and the second condition of \eqref{eq:DistN_kgamma} to \eqref{eq:dV_eta_xi_zeta} and by using \eqref{eq:ubd_V}, it is obtained that
\begin{align*}
	\dot V&\le -k_s\frac{2-\sqrt 2}{2}(|\eta|^2+|\xi|^2+|\tilde \nu|^2)
	\le -\frac{k_s(2-\sqrt{2})}{12N+15+\sqrt 5}V.
\end{align*}
Finally, the condition \eqref{eq:DistN_kgamma} yields $\dot V\leq -2\lambda_s V$.

\section*{acknowledgement}

Need for Section \ref{sec:dwell} is motivated by the associate editor, for which the authors are grateful. The authors also appreciate the anonymous reviewer for the improved Algorithm \ref{algo2}.


\begin{IEEEbiographynophoto}{Taekyoo Kim}
	received his B.S. degree in mechanical engineering and M.S., and Ph.D. degrees in electrical engineering from Seoul National University, Korea, in 2007, 2009, and 2019, respectively.
	From 2010 to 2014, he was a researcher at Agency for Defense Development, Korea.
	He is currently a postdoctoral researcher at Education and Research program for Future ICT Pioneers, Seoul National University, Korea.
	His research interests include distributed estimation/control, output regulation, and geometric control technique.
\end{IEEEbiographynophoto}
\begin{IEEEbiographynophoto}{Donggil Lee}
	received the B.S. degree in electrical engineering from Seoul National University in 2015.	
	He is currently pursuing a Ph.D. degree in electrical engineering from Seoul National University.	
	His current research interests are in the area of distributed control, encrypted control, and model predictive control.	
\end{IEEEbiographynophoto}
\begin{IEEEbiographynophoto}{Hyungbo Shim}
	received the B.S., M.S., and Ph.D. degrees from Seoul National University, Korea, and held the post-doc position at University of California, Santa Barbara till 2001. 
	He joined Hanyang University, Seoul, Korea, in 2002. 
	Since 2003, he has been with Seoul National University, Korea. 
	He served as associate editor for Automatica, IEEE Trans. on Automatic Control, Int. Journal of Robust and Nonlinear Control, and European Journal of Control, and as editor for Int. Journal of Control, Automation, and Systems. 
	He was the Program Chair of ICCAS 2014 and Vice-program Chair of IFAC World Congress 2008. 
	His research interest includes stability analysis of nonlinear systems, observer design, disturbance observer technique, secure control systems, and synchronization.
\end{IEEEbiographynophoto}

\end{document}